\documentclass[11pt]{article}
\pdfoutput=1
\usepackage{mdframed}
\usepackage{tikz,tkz-graph}
\usetikzlibrary{arrows,shapes}
\usetikzlibrary{calc}
\tikzset{>= latex}
\tikzstyle{vertex}=[circle, draw,fill=gray!20, inner sep=0pt, minimum size=16pt]
\tikzstyle{svertex}=[circle, draw,fill=gray!20, inner sep=0pt, minimum size=8pt]

\def\showauthornotes{0}
\def\showtableofcontents{0}
\def\showkeys{0}
\def\showdraftbox{0}
\def\showcolorlinks{1}
\def\usemicrotype{1}
\def\showfixme{0}

\usepackage{etex}

\usepackage[utf8]{inputenc}

\usepackage{xspace,enumitem}

\usepackage[]{xcolor}

\usepackage[T1]{fontenc}
\usepackage[full]{textcomp}

\usepackage[american]{babel}

\usepackage{mathtools}

\usepackage{bm}

\usepackage{amsthm}

\newtheorem{theorem}{Theorem}[section]
\newtheorem*{theorem*}{Theorem}

\newtheorem*{proposition*}{Proposition}
\newtheorem{lemma}[theorem]{Lemma}
\newtheorem*{lemma*}{Lemma}

\newtheorem*{conjecture*}{Conjecture}

\newtheorem*{fact*}{Fact}

\newtheorem*{hypothesis*}{Hypothesis}

\theoremstyle{definition}

\newtheorem{example}[theorem]{Example}

\newtheorem{openquestion}[theorem]{Open Question}

\theoremstyle{remark}
\newtheorem{claim}[theorem]{Claim}
\newtheorem*{claim*}{Claim}
\newtheorem{remark}[theorem]{Remark}
\newtheorem*{remark*}{Remark}

\newtheorem*{observation*}{Observation}

\usepackage[letterpaper,
	top=1.2in,
	bottom=1.2in,
	left=1.4in,
	right=1.4in]{geometry}

\usepackage[varg]{pxfonts} 

\ifnum\showkeys=1
\usepackage[color]{showkeys}
\fi

\definecolor{OliveGreen}{rgb}{0,0.6,0}
\ifnum\showcolorlinks=1
\usepackage[
pagebackref,
colorlinks=true,
urlcolor=blue,
linkcolor=blue,
citecolor=OliveGreen,
]{hyperref}
\fi

\ifnum\showcolorlinks=0
\usepackage[
pagebackref,
colorlinks=false,
pdfborder={0 0 0}
]{hyperref}
\fi

\usepackage{prettyref}

\newcommand{\savehyperref}[2]{\texorpdfstring{\hyperref[#1]{#2}}{#2}}

\newrefformat{eq}{\savehyperref{#1}{\textup{(\ref*{#1})}}}
\newrefformat{lem}{\savehyperref{#1}{Lemma~\ref*{#1}}}
\newrefformat{def}{\savehyperref{#1}{Definition~\ref*{#1}}}
\newrefformat{thm}{\savehyperref{#1}{Theorem~\ref*{#1}}}
\newrefformat{cor}{\savehyperref{#1}{Corollary~\ref*{#1}}}
\newrefformat{cha}{\savehyperref{#1}{Chapter~\ref*{#1}}}
\newrefformat{sec}{\savehyperref{#1}{Section~\ref*{#1}}}
\newrefformat{app}{\savehyperref{#1}{Appendix~\ref*{#1}}}
\newrefformat{tab}{\savehyperref{#1}{Table~\ref*{#1}}}
\newrefformat{fig}{\savehyperref{#1}{Figure~\ref*{#1}}}
\newrefformat{hyp}{\savehyperref{#1}{Hypothesis~\ref*{#1}}}
\newrefformat{alg}{\savehyperref{#1}{Algorithm~\ref*{#1}}}
\newrefformat{rem}{\savehyperref{#1}{Remark~\ref*{#1}}}
\newrefformat{item}{\savehyperref{#1}{Item~\ref*{#1}}}
\newrefformat{step}{\savehyperref{#1}{step~\ref*{#1}}}
\newrefformat{conj}{\savehyperref{#1}{Conjecture~\ref*{#1}}}
\newrefformat{fact}{\savehyperref{#1}{Fact~\ref*{#1}}}
\newrefformat{prop}{\savehyperref{#1}{Proposition~\ref*{#1}}}
\newrefformat{prob}{\savehyperref{#1}{Problem~\ref*{#1}}}
\newrefformat{claim}{\savehyperref{#1}{Claim~\ref*{#1}}}
\newrefformat{relax}{\savehyperref{#1}{Relaxation~\ref*{#1}}}
\newrefformat{red}{\savehyperref{#1}{Reduction~\ref*{#1}}}
\newrefformat{part}{\savehyperref{#1}{Part~\ref*{#1}}}

\newcommand{\Sref}[1]{\hyperref[#1]{\S\ref*{#1}}}

\usepackage{nicefrac}

\ifnum\usemicrotype=1
\usepackage{microtype}
\fi

\ifnum\showauthornotes=1
\newcommand{\Authornote}[2]{{\sffamily\small\color{red}{[#1: #2]}}}
\newcommand{\Authornotecolored}[3]{{\sffamily\small\color{#1}{[#2: #3]}}}
\newcommand{\Authorcomment}[2]{{\sffamily\small\color{gray}{[#1: #2]}}}
\newcommand{\Authorstartcomment}[1]{\sffamily\small\color{gray}[#1: }

\newcommand{\Authorfnote}[2]{\footnote{\color{red}{#1: #2}}}
\newcommand{\Authorfixme}[1]{\Authornote{#1}{\textbf{??}}}
\newcommand{\Authormarginmark}[1]{\marginpar{\textcolor{red}{\fbox{\Large #1:!}}}}
\else
\newcommand{\Authornote}[2]{}
\newcommand{\Authornotecolored}[3]{}
\newcommand{\Authorcomment}[2]{}
\newcommand{\Authorstartcomment}[1]{}

\newcommand{\Authorfnote}[2]{}
\newcommand{\Authorfixme}[1]{}
\newcommand{\Authormarginmark}[1]{}
\fi

\ifnum\showfixme=0

\fi

\usepackage{boxedminipage}

 \usepackage{dsfont}
\usepackage{mathrsfs}

\newcommand{\textparen}[1]{\text{(#1)}}

\ifx\because\undefined
\newcommand{\because}[1]{\textparen{because #1}}
\else
\renewcommand{\because}[1]{\textparen{because #1}}
\fi

\newcommand\bdot\bullet

\DeclareMathOperator{\lb}{lb}

\DeclareMathOperator{\LP}{LP}
\DeclareMathOperator{\OPT}{OPT}

\newcommand{\R}{\mathbb R}

\newcommand{\cA}{\mathcal A}

\newcommand{\cC}{\mathcal C}

\newcommand{\cF}{\mathcal F}

\renewcommand{\leq}{\leqslant}
\renewcommand{\le}{\leqslant}
\renewcommand{\geq}{\geqslant}

\ifnum\showdraftbox=1
\newcommand{\draftbox}{\begin{center}
  \fbox{%
    \begin{minipage}{2in}%
      \begin{center}%
          \Large\textsc{Working Draft}\\%
        Please do not distribute%
      \end{center}%
    \end{minipage}%
  }%
\end{center}
\vspace{0.2cm}}
\else
\newcommand{\draftbox}{}
\fi

\let\epsilon=\varepsilon

\numberwithin{equation}{section}

\newcommand{\MYstore}[2]{%
  \global\expandafter \def \csname MYMEMORY #1 \endcsname{#2}%
}

\newcommand{\MYload}[1]{%
  \csname MYMEMORY #1 \endcsname%
}

\newcommand{\MYnewlabel}[1]{%
  \newcommand\MYcurrentlabel{#1}%
  \MYoldlabel{#1}%
}

\newcommand{\MYdummylabel}[1]{}

\newcommand{\torestate}[1]{%
  \let\MYoldlabel\label%
  \let\label\MYnewlabel%
  #1%
  \MYstore{\MYcurrentlabel}{#1}%
  \let\label\MYoldlabel%
}

\newcommand{\restatetheorem}[1]{%
  \let\MYoldlabel\label
  \let\label\MYdummylabel
  \begin{theorem*}[Restatement of \prettyref{#1}]
    \MYload{#1}
  \end{theorem*}
  \let\label\MYoldlabel
}

\newcommand{\restatelemma}[1]{%
  \let\MYoldlabel\label
  \let\label\MYdummylabel
  \begin{lemma*}[Restatement of \prettyref{#1}]
    \MYload{#1}
  \end{lemma*}
  \let\label\MYoldlabel
}

\newcommand{\restateprop}[1]{%
  \let\MYoldlabel\label
  \let\label\MYdummylabel
  \begin{proposition*}[Restatement of \prettyref{#1}]
    \MYload{#1}
  \end{proposition*}
  \let\label\MYoldlabel
}

\newcommand{\restatefact}[1]{%
  \let\MYoldlabel\label
  \let\label\MYdummylabel
  \begin{fact*}[Restatement of \prettyref{#1}]
    \MYload{#1}
  \end{fact*}
  \let\label\MYoldlabel
}

\newcommand{\restate}[1]{%
  \let\MYoldlabel\label
  \let\label\MYdummylabel
  \MYload{#1}
  \let\label\MYoldlabel
}

\let\origparagraph\paragraph
\renewcommand{\paragraph}[1]{\origparagraph{#1.}}

\allowdisplaybreaks

\DeclareMathOperator*{\con}{\cC}
\DeclareMathOperator*{\out}{\mathrm{O}}
\DeclareMathOperator*{\inn}{\mathrm{I}}

\DeclareMathOperator*{\low}{\mathrm{low}}

\newcommand{\hide}[1]{}

\title{\bfseries Approximating ATSP by Relaxing Connectivity}

\author{%
Ola Svensson \\
EPFL\\
\textit{ola.svensson@epfl.ch}
}

\begin{document}

\maketitle
\draftbox
\thispagestyle{empty}

\begin{abstract}
  The standard LP relaxation of the asymmetric traveling salesman problem has
  been conjectured to have a constant integrality gap in the metric case. We
  prove this conjecture when restricted to shortest path metrics of
  node-weighted digraphs. Our arguments are constructive and give a constant
  factor approximation algorithm for these metrics.  We remark that the
  considered case is more general than the directed analog of the special case of the symmetric traveling salesman problem
   for which there were recent improvements on Christofides' algorithm.

  The main idea of our approach is to first consider an easier problem obtained
  by significantly relaxing the general connectivity requirements into local
  connectivity conditions.  For this relaxed problem, it is  quite easy to give
  an algorithm with a guarantee of $3$ on node-weighted shortest path metrics.
  More surprisingly, we then show that \emph{any} algorithm (irrespective of
  the metric) for the relaxed problem can be turned into an algorithm for the
  asymmetric traveling salesman problem  by only losing a small constant factor
  in the performance guarantee. This leaves open the intriguing task of 
  designing a ``good'' algorithm for the relaxed problem on general metrics. 

\end{abstract}


\clearpage

\ifnum\showtableofcontents=1
{
\tableofcontents
\thispagestyle{empty}
 }
\fi

\clearpage

\section{Introduction}
The traveling salesman problem is one of the most fundamental combinatorial
optimization problems. Given a set $V$ of $n$ cities and a distance/weight
function $w: V\times V \rightarrow \R^+$, it is the problem  of finding a tour
of minimum total weight that visits each city exactly once.  There are two
variants of this general definition: the \emph{symmetric} traveling
salesman problem (STSP) and the \emph{asymmetric} traveling salesman problem (ATSP).
In the symmetric version we assume $w(u,v)=w(v,u)$ for each pair $u,v\in V$ of
cities; whereas we make no such assumption in the more general asymmetric
traveling salesman problem.

In both versions, it is common to assume the triangle inequality and we shall
do so in the rest of this paper. Recall that the triangle inequality says that
for any triple $i,j, k$ of cities, we have $w(i,j) + w(j,k) \geq w(i,k)$.  In
other words, it is not more expensive to take the direct path compared to
a path that makes a detour.  Another equivalent view of the triangle inequality
is that, instead of insisting that each city is visited exactly once, we should
find a tour that visits each city \emph{at least} once.  These assumptions are
arguably natural in many, if not most, settings. They are also necessary in the
following sense: any reasonable approximation algorithm (with approximation
guarantee $O(\exp(n))$) for the traveling salesman problem without the triangle
inequality would imply $P=NP$ because it would solve the problem of deciding
Hamiltonicity.

Understanding the approximability of the symmetric and the asymmetric traveling
salesman problem (where we have the triangle inequality) turns out to be
a much more interesting and notorious problem.
On the one hand, the strongest known  inapproximability results,
by Karpinski, Lampis, and Schmied~\cite{KarpinskiLS13}, say that it is NP-hard
to approximate STSP within a factor of $123/122$ and that it is NP-hard to
approximate ATSP within a factor of $75/74$.  On the other hand, the current
best approximation algorithms are far from these guarantees, especially in the
case of ATSP.

For the symmetric traveling salesman problem, Christofides' beautiful algorithm
from 1976 sill achieves the best known approximation guarantee of $1.5$~\cite{Ch76}.
However, a recent series of papers~\cite{GharanSS11,MomkeS11,Mucha12,SeboV14},
broke this barrier for the interesting special case of  shortest path metrics
of unweighted undirected graphs\footnote{The shortest path  metric of a graph $G= (V,E)$ is defined 
as follows: the weight $w(u,v)$ between cities
$u,v\in V$ equals the shortest path between $u$ and $v$ in $G$. If the graph is
node-weighted $f: V \rightarrow \R^+$, the weight/length  of an edge $\{u,v\}\in E$  is $f(u) + f(v)$.}.
Specifically, Oveis Gharan, Saberi, and Singh~\cite{GharanSS11} first gave an
approximation guarantee of $1.5-\epsilon$; M\"omke and 
Svensson~\cite{MomkeS11} proposed a different approach yielding
a $1.461$-approximation guarantee; Mucha~\cite{Mucha12} gave a tighter analysis
of this algorithm; and Seb\"{o} and Vygen~\cite{SeboV14} significantly
developed the approach to give the current best approximation guarantee of
$1.4$.

The interest in shortest path metrics has several motivations. It is a natural
special case that seems to capture the difficulty of the problem: it remains APX-hard
and the worst known integrality gap for the Held-Karp relaxation is of this
type. Moreover, it has an attractive graph theoretic formulation: given an
unweighted graph,  find a shortest tour that visits each vertex at least once.
This is the (unweighted) ``graph'' analog of STSP. Indeed, if allow the graph
to be edge-weighted, this formulation is equivalent to STSP on general metrics.
Let us also mention that the polynomial time approximation scheme for the
symmetric traveling salesman problem on  planar graphs was first obtained for
the special case of unweighted graphs~\cite{GrigniKP95}, i.e., when restricted
to shortest path metrics of unweighted graphs,  and then generalized to the
case of edge-weights~\cite{AroraGKKW98}. For STSP, it remains  a major open
problem whether the ideas in~\cite{GharanSS11,MomkeS11,Mucha12,SeboV14} can be
applied to general metrics.  We further discuss this in
Section~\ref{sec:discuss}.

The gap in our understanding is much larger for the asymmetric
traveling salesman problem for which  it remains a notorious open problem to design
an algorithm with \emph{any} constant approximation guarantee. This is a particularly intriguing as
the standard linear programming relaxation, often referred to as the Held-Karp
relaxation, is only known to have an integrality gap of at least $2$~\cite{CharikarGK06}. There are in
general  two available approaches for designing approximation algorithms for
ATSP in the literature. The first approach is due to Frieze, Galbiati, and
Maffiolo~\cite{FriezeGM82} who gave a $\log_2(n)$-approximation
algorithm for ATSP already in 1982. Their basic idea is
simple and elegant: a minimum weight cycle cover has weight at most that of an optimal tour
and it will decrease the number of connected components by a factor of at least
$2$. Hence, if we repeat the selection of a minimum weight cycle cover
$\log_2(n)$ times, we get a connected Eulerian graph which (by shortcutting) is
a $\log_2(n)$-approximate tour. Although the above analysis is tight only
in the case when almost all cycles in the cycle covers have length $2$, it is
highly non-trivial to refine the method to decrease the number of iterations.
It was first in 2003 that Bl\"{a}ser~\cite{Blaser08} managed to give an
approximation guarantee of $0.999\log_2(n)$. This was improved shortly
thereafter by  Kaplan, Lewenstein, Shafrir and Sviridenko~\cite{KaplanLS05} who
further developed this approach to obtain a $4/3\log_3(n) \approx 0.84
\log_2(n)$-approximation algorithm; and later by Feige and Singh~\cite{FeigeS07} who obtained an approximation guarantee of $2/3\log_2(n)$.

A second approach was more recently proposed in an influential and beautiful paper by
Asadpour, Goemans, Madry, Oveis Gharan, and Saberi~\cite{AsadpourGMGS10} who
gave an $O(\log n/\log \log n)$-approximation algorithm for ATSP.
Their approach is based on finding a so-called $\alpha$-thin spanning tree
which is a (unweighted) graph theoretic problem.  Here, the parameter $\alpha$
is proportional to the approximation guarantee so $\alpha = O(\log n/\log \log
n)$ in~\cite{AsadpourGMGS10}.  Following their publication, Oveis Gharan and
Saberi~\cite{GharanS11} gave an efficient algorithm for finding $O(1)$-thin spanning trees for
planar and bounded genus graphs yielding a constant factor approximation
algorithm for ATSP on these graph classes.  Also, in a very recent major
progress, Anari and Oveis Gharan~\cite{AnariG14} showed the existence of  $O(\textrm{polylog} \log n)$-thin
spanning trees for general instances. This implies a $O(\textrm{polylog} \log n)$ upper
bound on the integrality gap  of the Held-Karp
relaxation.  Hence, it gives an efficient so-called estimation algorithm for
estimating the optimal value of a tour within a factor $O(\textrm{polylog} \log n)$ but, as
their arguments are non-constructive,  no approximation algorithm for finding
a tour of matching guarantee. The result
in~\cite{AnariG14} is based on developing and extending several advanced techniques.
Notably, they rely on their extension~\cite{AnariG14a} of the recent proof of the
Kadison-Singer conjecture which was a  major breakthrough  by Marcus, Spielman,
and Srivastava~\cite{MSS13}.

To summarize, the current best approximation algorithm has a guarantee of
$O(\log n/\log\log n)$~\cite{AsadpourGMGS10} and the best upper bound on the
integrality gap of the Held-Karp relaxation is $O(\textrm{polylog}
\log n)$~\cite{AnariG14}. These two bounds are far away from the known
inapproximability results~\cite{KarpinskiLS13} and from the lower bound of $2$
on the integrality gap of the Held-Karp relaxation~\cite{CharikarGK06}.
Moreover, there were no better approximation algorithms known  in the case
of shortest path metrics of unweighted digraphs for which there was recent
progress in the undirected setting. In particular, it is not clear how to use
the two available approaches mentioned above to get an improved approximation guarantee  in this case:
in the cycle cover approach, the main difficulty is to bound the number of
iterations  and, in the thin spanning tree approach, ATSP is reduced to an
unweighted graph theoretic problem.

\subsection{Our Results and Overview of Approach}
We propose a new approach for approximating the asymmetric traveling salesman
problem based on relaxing the global connectivity constraints into local
connectivity conditions.   We also use this approach to obtain the following
result where we refer to ATSP on shortest path metrics of node-weighted digraphs
as Node-Weighted ATSP.
\begin{theorem}
  \label{thm:constant}
  There is a constant approximation algorithm for Node-Weighted ATSP.
  Specifically, for Node-Weighted ATSP, the integrality gap of the Held-Karp
  relaxation is at most $15$ and, for any $\epsilon>0$, there is a polynomial
  time algorithm that finds a tour of weight at most $(27+\epsilon)\OPT_{HK}$
  where $\OPT_{HK}$ denotes the optimal value of the Held-Karp relaxation.
\end{theorem}
As further discussed in Section~\ref{sec:discuss}, the constants in the theorem
can be slightly improved by specializing our general approach to the
node-weighted case. However, it remains an interesting open problem to give
a tight bound on the integrality gap.

Let us continue with a brief overview of our approach that is not restricted to
the node-weighted version. It is illustrative to consider the following
``naive'' algorithm that actually was the starting point of this work:
\begin{enumerate}
  \item Select a random cycle cover $C$ using the Held-Karp relaxation.
    
    It is well known that one can  sample such a cycle cover $C$ of   expected
    weight equal to the optimal value $\OPT_{HK}$  of the
    Held-Karp relaxation.
  \item While there exist more than one component, add the lightest cycle (i.e., the cycle of smallest weight) that
    decreases the number of components. 
\end{enumerate}
It is clear that the above algorithm always returns a solution to ATSP: we
start with a Eulerian graph\footnote{Recall that a directed graph is Eulerian if the in-degree equals the out-degree of each vertex.} and the graph stays Eulerian during the execution
of the while-loop which does not terminate until the graph is connected. This
gives a tour that visits each vertex at least once and hence a solution to ATSP
(using that we have the triangle-inequality). 
However, what is the weight of the obtained tour? First, as remarked above, we
have that the expected weight of the cycle cover is $\OPT_{HK}$. So if $C$
contains $k =|C|$ cycles, we would expect that a cycle in $C$ has weight
$\OPT_{HK}/k$ (at least on average).  Moreover,  the number of cycles added in
Step~$2$ is at most $k-1$ since each cycle decreases the number of components
by at least one.  Thus, if each cycle in Step~$2$ has weight at most  the average
weight $\OPT_{HK}/k$ of a cycle in $C$,  we obtain a $2$-approximate tour
of weight at most $\OPT_{HK} + \frac{k-1}{k} \OPT_{HK} \leq 2 \OPT_{HK}$.

Unfortunately, it seems hard to find a cycle cover $C$ so that we can always
connect it with light cycles. Instead, what we can do, is to first select
a cycle cover $C$ then add light cycles that decreases the number of components
as long as possible.  When there are no more light cycles to add, the
vertices/cities are partitioned into $V_1, \ldots, V_k$ connected components.
In order to make progress from this point, we would like to find a ``light''
Eulerian set $F$ of edges that crosses  the cuts $\{(V_i, \bar V_i)\mid i=1,2,
\ldots, k\}$. We could then hope to add $F$ to our solution and continue from
there. It turns out that it is very important what ``light'' means in this
context.  For our arguments to work, we need that $F$ is selected so that  the
weight of the edges in each component has weight at most $\alpha$ times what
the linear programming solution ``pays'' for the vertices in that component. This is the intuition behind
the definitions in Section~\ref{sec:LCdef} of Local-Connectivity ATSP and
$\alpha$-light algorithms for that problem. We also need to be very careful in
which way we add edges from light cycles and how to use the $\alpha$-light
algorithm for Local-Connectivity ATSP. In Section~\ref{sec:LocalToGlobal}, our
algorithm will iteratively solve the Local-Connectivity ATSP and, in each
iteration, it will  add a carefully chosen subset of the found edges together with
light cycles.

We remark that in Local-Connectivity ATSP we have relaxed the global
connectivity properties of ATSP into local connectivity conditions that only
say that we need to find a Eulerian set of edges that crosses at most $n=|V|$
cuts defined  by a partitioning of the vertices. In spite of that, we are able
to leverage the intuition above to obtain our main technical result:
\begin{theorem*}[Simplified statement of Theorem~\ref{thm:LoctoGlo}]
  The integrality gap of the Held-Karp relaxation is at most $5\alpha$ if there exists an $\alpha$-light algorithm $\cA$ for Local-Connectivity ATSP. Moreover, for any $\epsilon > 0$, we can  find a $(9+\epsilon) \alpha$-approximate tour in time polynomial in $n, 1/\epsilon$, and in the running time of $\cA$.
\end{theorem*}
The proof  of the above theorem (Section~\ref{sec:LocalToGlobal}) is based on
generalizing  and, as alluded to above, deviating from the above intuition in several ways. First, we  start with
a carefully chosen ``Eulerian partition'' which  generalizes the role of the
cycle cover $C$ in Step~1 above. Second, both the iterative use of the $\alpha$-light
algorithm for Local-Connectivity ATSP and the way we add light cycles  are 
done  in a careful and dependent manner so as to be able to bound the total
weight of the returned solution.
Theorem~\ref{thm:constant} follows from Theorem~\ref{thm:LoctoGlo} together
with  a $3$-light algorithm for Node-Weighted Local-Connectivity ATSP. The
$3$-light algorithm, described in Section~\ref{sec:lcATSPalg}, is a rather
simple application of classic theory of flows and circulations. We also remark
that it is the only part of the paper that relies on having shortest path
metrics of node-weighted digraphs. 

Our work raises several natural questions. Perhaps the most immediate and
intriguing question is whether there is a $O(1)$-light algorithm for
Local-Connectivity ATSP on general metrics. We further elaborate on this and
other related questions in Section~\ref{sec:discuss}.

\section{Preliminaries}

\subsection{Basic Notation}

Consider a directed graph $G=(V,E)$. For a subset $S \subseteq V$, we let $\delta^+(S)
= \{(u,v) \in E\mid u\in S, v\not \in S\}$ be the outgoing edges and  we let $\delta^-(S) = \{(u,v)\in E\mid
u\not \in S, v\in S\}$ be the incoming edges of the cut defined by $S$. When
considering a subset $E'\subseteq E$ of the edges, we denote the restrictions
to that subset by $\delta^+_{E'}(S) = \delta^+(S) \cap E'$ and by
$\delta^-_{E'}(S) = \delta^-(S) \cap E'$. We also let
$\con(E') = \{\tilde G_1 = (\tilde V_1, E_1), \tilde G_2 = (\tilde V_2,\tilde E_2), \ldots,\tilde G_k = (\tilde V_k,\tilde E_k)\}$
denote the set of subgraphs corresponding to the $k$ connected components of
the graph $(V,E')$; the vertex set $V$ will always be clear from the context.
Here connected means that the subgraphs are connected if we undirect the edges. 

When considering a function $f: U \rightarrow \R$, we let $f(X) = \sum_{x\in X}
f(x)$ for  $X\subseteq U$. For example, if $G$ is edge weighted, i.e., there
exists a function $w: E \rightarrow \R$, then $w(E')$ denotes the total weight
of the edges in $E' \subseteq E$. Similarly, if $G$ is node-weighted, then
there exists a function $f: V \rightarrow  \R$ and $f(S)$ denotes the total
weight of the vertices in $S \subseteq V$. When talking about graphs, we shall
slightly abuse notation and sometimes write $w(G)$ instead of $w(E)$ and $f(G)$
instead of $f(V)$ when it is clear from the context that $w$ and $f$ are
functions on the edges and vertices. Finally, our \emph{subsets of
edges are multisets}, i.e., may contain the same edge several times. The set
operators $\cup, \cap, \setminus$ are defined in the natural way. For example,
$\{e_1, e_1, e_2\} \cup \{e_1, e_2\}  = \{e_1,e_1,e_1, e_2, e_2\}, \{e_1,e_1,
e_2\} \cap \{e_1, e_2\} = \{e_1, e_2\}$, and $\{e_1, e_1, e_2\} \setminus
\{e_1, e_2\} = \{e_1\}$. Other sets, such as subsets of vertices, will always
be simple sets without any multiplicities.

\subsection{The (Node-Weighted) Asymmetric Traveling Salesman Problem}

It will be convenient to define ATSP using the Eulerian point of view, i.e., we
wish to find a tour that visits each vertex at least once. As already mentioned
in the introduction, this definition is equivalent to that of visiting each city exactly once (in the metric completion) since we assume the triangle
inequality.
\begin{mdframed}
\begin{center} \textbf{ATSP} \end{center}
\begin{description}
  \item[\textnormal{\emph{Given:}}]An edge-weighted (strongly connected) digraph $G=(V,E, w: E \rightarrow \R^+)$.  
	
\item[\textnormal{\emph{Find:}}]A connected Eulerian digraph $G' = (V,E')$
  where $E'$ is a multisubset of $E$ that 
  minimizes $w(E')$. 
\end{description}
\end{mdframed}
\vspace{2mm}

Similar to the  recent progress on STSP, it is natural to consider
special cases that are easier to argue about but at the same time 
capture the combinatorial structure of the problem. In particular, we shall consider the
\emph{Node-Weighted} ATSP, where  we assume that there exists a weight function
$f: V\rightarrow R^+$ on the vertices so that $w(u,v) = f(u)$.  (Another
equivalent definition, which also applies to undirected graphs, is to let the
weight of an edge $(u,v)$ be $f(u)+f(v)$. This is equivalent to the definition
above, if scaled down by a factor of $2$, since the solutions are Eulerian.)

Note that this generalizes ATSP on shortest path metrics of unweighted
digraphs: that is the problem where $f$ is the constant function.  As
a curiosity, we also note that the recent progress on STSP when restricted to
shortest path metrics of unweighted  graphs is not known to generalize to the
node-weighted case. We raise this as an interesting open problem in
Section~\ref{sec:discuss}.

\subsection{Held-Karp Relaxation}

The Held-Karp relaxation has a variable $x_e \geq 0$ for every edge in the given edge-weighted graph $G = (V,E,w)$. The intended solution is that $x_e$ should equal the number of times $e$ is used in the solution. The relaxation $\LP(G)$ is now defined as follows:

\begin{align*}
\arraycolsep=1.4pt\def\arraystretch{1.2}
\begin{array}{lrlr}
\mbox{minimize} \qquad & \displaystyle \sum_{e\in E} x_e w(e) \\[7mm]
\mbox{subject to} \qquad  & \displaystyle x(\delta^+(v)) = & \displaystyle x(\delta^-(v)) & v\in V, \\  
& \displaystyle x(\delta^+(S)) \geq & 1 & \emptyset \neq S \subset V, \\
& x \geq & 0.
\end{array}
\end{align*}
The first set of constraints says that the in-degree should equal the out-degree for each vertex, i.e., the solution should be Eulerian. The second set of constraints enforces that the solution is connected and they are sometimes referred to as subtour elimination constraints.
For Node-Weighted ATSP, we can write the objective function of the linear program as $\sum_{v\in V} f(v) \cdot x(\delta^+(v))$, where $f: V \rightarrow \R^+$ is the weights  on the vertices that defines the node-weighted metric.

Finally, we remark that although the Held-Karp relaxation has exponentially many constraints, it is well known that we can solve it in polynomial time either by using the ellipsoid method with a separation oracle or by formulating an equivalent compact (polynomial size) linear program.

\section{ATSP with Local Connectivity}
\label{sec:LCdef}

In this section we define a seemingly  easier problem than ATSP by relaxing the
connectivity requirements.  
Consider an optimal solution $x^*$ to $\LP(G)$. Its value, which is a lower
bound on $\OPT$,  can be decomposed into a  ``lower bound'' for each vertex $v$:  
\begin{align*}
  \sum_{e\in E} x^*_e w(e) = \sum_{v\in V}  \underbrace{\sum_{e\in \delta^+(v)} x^*_e w(e).}_{\mbox{lower bound for $v$}}
\end{align*}

With this intuition,  we let $\lb: V \rightarrow \R$ be the lower bound function defined by $\lb_{x^*, G}(v) = \sum_{e\in \delta^+(v)} x^*_e w(e)$. 
We simplify notation and write $\lb$ instead of $\lb_{x^*, G}$ as $G$  will always be clear from the context and therefore also $x^*$ (if the optimal solution to $\LP(G)$ is not unique then make an arbitrary but consistent choice). Note that $\lb(V)$ equals the value of the optimal solution to the Held-Karp relaxation.

Perhaps the main difficulty of ATSP is to satisfy the connectivity requirement,
i.e., to select a Eulerian subset $F$ of edges that connects the whole graph. 
We shall now relax this condition to obtain what we call
\emph{Local-Connectivity} ATSP:
\newpage
\begin{mdframed}
\begin{center} \textbf{Local-Connectivity ATSP} \end{center}
\begin{description}
  \item[\textnormal{\emph{Given:}}]An edge-weighted (strongly connected)
	digraph $G=(V,E, w)$ and   a partitioning $V_1 \cup V_2 \cup \ldots \cup V_k$
	of the vertices that satisfy:  the graph induced by $V_i$ is strongly
	connected for $i=1,\ldots, k$.   
	
\item[\textnormal{\emph{Find:}}]A Eulerian multisubset $F$ of $E$   such
that
\begin{align*}
  |\delta^+_{F}(V_i)| \geq 1 \mbox{ for } i=1, 2, \ldots, k \qquad \mbox{and}
  \qquad \max_{\tilde G\in \con(F)} \frac{w(\tilde G)}{\lb(\tilde G)} \mbox{ is minimized.}
\end{align*}
\end{description}
\end{mdframed}
\vspace{2mm}
Recall that $\con(F)$ denotes the set of connected components of the graph
$(V,F)$. We remark that the restriction that each $V_i$ should induce a
strongly connected component is not necessary but it makes our proofs in
Section~\ref{sec:lcATSPalg} easier.

We say that an algorithm for Local-Connectivity ATSP is 
\emph{$\alpha$-light} if it is guaranteed (over all instances) to find a solution $F$
such that 
\begin{align}
  \label{eq:approx}
  \max_{\tilde G\in \con(F)} \frac{w(\tilde G)}{\lb(\tilde G)} \leq
\alpha.
\end{align}
We also say that an algorithm is $\alpha$-light on an ATSP instance
$G=(V,E,w)$ if, for each partitioning $V_1 \cup \ldots \cup V_k$ of $V$ (such
that $V_i$ induces a strongly connected graph), it returns a solution
satisfying~\eqref{eq:approx}. We remark that we use the $\alpha$-light
terminology to avoid any ambiguities with the concept of approximation
algorithms because  an $\alpha$-light algorithm does not compare its solution
with respect to an optimal solution  to the given instance of Local-Connectivity ATSP. 

An $\alpha$-approximation algorithm for ATSP with respect to the Held-Karp
relaxation is trivially an $\alpha$-light algorithm for Local-Connectivity
ATSP: output the same Eulerian subset $F$ as the algorithm for ATSP.  Since
the set $F$ connects the graph we have $\max_{\tilde G \in \con(F)} w(\tilde
G)/\lb(\tilde G) = w(F)/\lb(V) \leq \alpha$.  Moreover, {Local-Connectivity}
ATSP seems like a significantly easier problem than ATSP as the Eulerian set of
edges only needs to cross $k$ cuts formed  by a partitioning of the vertices.
We substantiate this intuition by proving,  in Section~\ref{sec:lcATSPalg},
that there exists a simple $3$-approximation for {Local-Connectivity} ATSP on shortest path metrics of 
node-weighted graphs. We refer to this case as Node-Weighted Local-Connectivity ATSP. 
Perhaps more surprisingly, we show in Section~\ref{sec:LocalToGlobal} that any  
$\alpha$-light algorithm for Local-Connectivity ATSP  can be turned into an
algorithm for ATSP with an approximation guarantee of $5\alpha$
with respect to the same  lower bound (from the Held-Karp relaxation).

\begin{remark}
  \label{rem:lb}
  Our generic reduction from ATSP to Local-Connectivity ATSP (Theorem~\ref{thm:LoctoGlo}) is robust with respect to
  the definition of $\lb$ and there are many possibilities to define such
  a lower bound. Another natural example is $\lb(v) = \sum_{e\in \delta^+(v)
  \cup \delta^-(v)} x^*_e w(e)/2$. In fact, in order to get a constant bound on
  the integrality gap of the Held-Karp relaxation, our results say that it is
  enough to find an $O(1)$-light algorithm for
  Local-Connectivity ATSP with respect to  some nonnegative $\lb$ that only needs to
  satisfy that $\lb(V)$ is at most the value of the optimal solution to the LP.
  Even more generally,  if $\lb(V)$ is at most the value of an optimal tour
  then our methods would give a similar approximation guarantee (but not with
  respect to the Held-Karp relaxation). 
\end{remark}

\section{Approximating Local-Connectivity ATSP}

\label{sec:lcATSPalg}

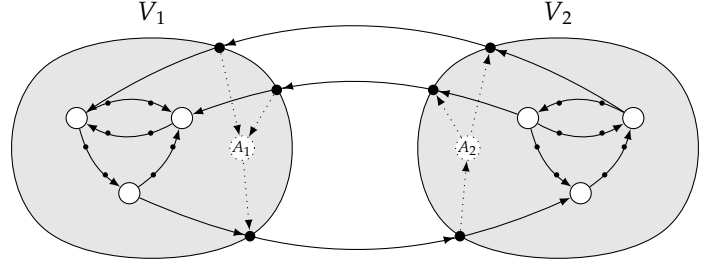
\begin{figure}[t]
\centering
\begin{tikzpicture}[scale=1]
\draw[fill=gray!20,xshift=0.3cm] plot [smooth cycle,tension=0.8] coordinates { (-4.4,-1.1) 
  (-1.6,-1.1) (-1.6,1.1) (-4.4,1.1)};
  \node at (-2.7,1.8) {\small $V_1$};

\node[svertex,fill=white] (v1) at (-2.3,0.4) {};
\node[svertex,fill=white] (v2) at (-3.7,0.4) {};
\node[svertex,fill=white] (v3) at (-3,-0.6) {};
\node[svertex, dotted, fill=white] (a1) at (-1.5, 0) {\tiny $A_1$};

\draw[fill=gray!20, xshift=-0.3cm] plot [smooth cycle,tension=0.8] coordinates { (1.6,-1.1) 
  (4.4,-1.1) (4.4,1.1) (1.6,1.1)};
  \node at (2.7,1.8) {\small $V_2$};

\node[svertex,fill=white] (v4) at (2.3,0.4) {};
\node[svertex,fill=white] (v5) at (3.7,0.4) {};
\node[svertex,fill=white] (v6) at (3,-0.6) {};
\node[svertex, dotted,fill=white] (a2) at (1.5, 0) {\tiny $A_2$};

\draw (v1) edge[->, bend left] node[near start,circle, fill=black, minimum size=2pt, inner sep=0pt]  {} node[near end,circle, fill=black, minimum size=2pt, inner sep=0pt] {} (v2);
\draw (v2) edge[->, bend left] node[near start,circle, fill=black, minimum size=2pt, inner sep=0pt]  {} node[near end,circle, fill=black, minimum size=2pt, inner sep=0pt] {}(v1);
\draw (v1) edge[<-, bend left=20] node[near start,circle, fill=black, minimum size=2pt, inner sep=0pt]  {} node[near end,circle, fill=black, minimum size=2pt, inner sep=0pt] {}(v3);
\draw (v2) edge[->, bend right=20] node[near start,circle, fill=black, minimum size=2pt, inner sep=0pt]  {} node[near end,circle, fill=black, minimum size=2pt, inner sep=0pt] {}(v3);

\draw (v4) edge[<-, bend left] node[near start,circle, fill=black, minimum size=2pt, inner sep=0pt]  {} node[near end,circle, fill=black, minimum size=2pt, inner sep=0pt] {}(v5);
\draw (v5) edge[<-, bend left] node[near start,circle, fill=black, minimum size=2pt, inner sep=0pt]  {} node[near end,circle, fill=black, minimum size=2pt, inner sep=0pt] {}(v4);
\draw (v4) edge[->, bend right=20] node[near start,circle, fill=black, minimum size=2pt, inner sep=0pt]  {} node[near end,circle, fill=black, minimum size=2pt, inner sep=0pt] {}(v6);
\draw (v5) edge[<-, bend left=20] node[near start,circle, fill=black, minimum size=2pt, inner sep=0pt]  {} node[near end,circle, fill=black, minimum size=2pt, inner sep=0pt] {}(v6);

\draw (v3) edge[->, bend right=24] node[near start,circle, fill=black, minimum size=4pt, inner sep=0pt] (t1) {} node[near end,circle, fill=black, minimum size=4pt, inner sep=0pt](t2) {} (v6);
\draw (v4) edge[bend right=20,->] node[near start,circle, fill=black, minimum size=4pt, inner sep=0pt] (t3) {} node[near end,circle, fill=black, minimum size=4pt, inner sep=0pt] (t4) {} (v1);
\draw (v5) edge[bend right=33, ->] node[near start,circle, fill=black, minimum size=4pt, inner sep=0pt](t5) {} node[near end,circle, fill=black, minimum size=4pt, inner sep=0pt] (t6) {} (v2);

\draw (t1) + (-0.08, 0.02) edge[->] (t1); 
\draw (t2) + (-0.08, -0.02) edge[->] (t2); 
\draw (t3) + (0.08, -0.02) edge[->] (t3); 
\draw (t4) + (0.08, 0.02) edge[->] (t4); 
\draw (t5) + (0.08, -0.03) edge[->] (t5); 
\draw (t6) + (0.08, 0.03) edge[->] (t6); 

\draw (t1) edge[<-, dotted] (a1);
\draw (t4) edge[->, dotted] (a1);
\draw (t6) edge[->, dotted] (a1);
\draw (t2) edge[->, dotted] (a2);
\draw (t3) edge[<-, dotted] (a2);
\draw (t5) edge[<-, dotted] (a2);
\end{tikzpicture}
\caption{A depiction of the construction of the auxiliary graph $G'$ (in the
  proof of Theorem~\ref{thm:lcapprox}):  edges are subdivided, an auxiliary
vertex $A_i$ is added for each partition $V_i$, and $A_i$ is ``connected'' to subdivisions of the
edges in $\delta^+(V_i)$ and $\delta^-(V_i)$.} 
\label{fig:Gaux}
\end{figure}
We give a simple $3$-light algorithm for Node-Weighted Local-Connectivity ATSP.
The proof is based on finding  an integral circulation that sends flow across
the cuts $\{(V_i, \bar V_i): i=1, 2, \ldots, k\}$ and, in addition, satisfies
that the outgoing flow of each vertex $v\in V$ is at most $\lceil
x^*(\delta^+(v))\rceil + 1$ which in turn, by the assumptions on the metric,  implies
a $3$-light algorithm.
\begin{theorem}
  \label{thm:lcapprox}
  There exists a  polynomial time $3$-light algorithm for Node-Weighted Local-Connectivity ATSP.
\end{theorem}
\begin{proof}
Let $G=(V,E,w)$ and $V_1 \cup V_2 \cup \ldots \cup
V_k$ be an instance of Local-Connectivity ATSP where $w: E \rightarrow \R^+
$ is a node-weighted metric defined by $f: V\rightarrow \R^+$. 
%
%
Let also $x^*$ be an optimal solution to $\LP(G)$. We prove the theorem by giving a polynomial time algorithm that finds a Eulerian multisubset $F$ of $E$ satisfying 
\begin{align}
  \label{eq:cycprop}
|\delta^+_{F}(V_i)|  \geq 1  \mbox{ for } i =1, \ldots, k \quad \mbox{and} \quad
|\delta^+_{F}(v)| & \leq \lceil x^*(\delta^+(v)) \rceil + 1 \mbox{ for } v\in V.
\end{align}
To see that this is sufficient, note that the Eulerian set $F$ forms a solution to the
Local-Connectivity ATSP instance because  $|\delta^+_{F}(V_i)|  \geq 1  \mbox{
for } i =1, \ldots, k$;  and it is  $3$-light since, for each $\tilde G
= (\tilde V, \tilde E) \in \con(F)$, we have (using that it is a node-weighted metric)
\begin{align*}
  \frac{w(\tilde G)}{\lb(\tilde G)}  
  = \frac{\sum_{v\in \tilde V}|\delta^+_{\tilde E}(v)| f(v)}{\sum_{v\in \tilde V} x^*(\delta^+(v)) f(v)} 
  \leq \frac{\sum_{v\in \tilde V} (\lceil x^*(\delta^+(v)) \rceil + 1) f(v)}{\sum_{v\in \tilde V} x^*(\delta^+(v)) f(v)}\leq 3. 
\end{align*}
The last inequality follows from $x^*(\delta^+(v)) \geq 1$ and
therefore $\lceil x^*(\delta^+(v)) \rceil + 1 \leq 3 x^*(\delta^+(v))$.

We proceed by describing a polynomial time algorithm for finding a Eulerian
set $F$ satisfying~\eqref{eq:cycprop}. We shall do so by finding a circulation
in an auxiliary graph $G'$ obtained from $G$ as follows (see also Figure~\ref{fig:Gaux}):

\begin{enumerate}
  \item Replace each edge $e=(u,v)$ in $G$ by adding vertices $\out_e, \inn_e$ and edges $(u, \out_e), (\out_e, \inn_e), (\inn_e, v)$; 
  \item For each partition $V_i$, $i =1, \ldots, k$, add an auxiliary vertex $A_i$ and edges $(A_i, \out_e)$ for every $e\in \delta^+(V_i)$ and $(\inn_e, A_i)$ for every $e\in \delta^-(V_i)$.
\end{enumerate}
Recall that a circulation in $G'$ is a  vector $y$ with a
nonnegative value for each edge satisfying flow conservation: $y(\delta^+(v)) =
y(\delta^-(v))$ for every vertex $v$.  The following claim follows from the
construction of $G'$ together with basic properties of flows
and circulations.
\begin{claim}
  We can in polynomial time find an integral circulation $y$ in $G'$ satisfying:
  $$
  y(\delta^+(A_i)) = 1 \mbox{ for } i=1,\ldots, k \quad \mbox{and} \quad y(\delta^+(v)) \leq \lceil x^*(\delta^+(v) \rceil \mbox{ for } v\in V.
  $$
\end{claim}
\begin{proof}
  We use the optimal solution $x^*$ to $\LP(G)$ to define a fractional circulation $y'$ in 
  $G'$ that satisfies the above degree bounds. As the vertex-degree bounds are integral, it
  follows  from basic facts about flows that we can in polynomial time find an
  integral circulation $y$ satisfying the same bounds (see e.g.
  Chapter~11~in~\cite{Schrijver03}). Circulation $y'$ is defined as follows:
  \begin{enumerate}
	\item for each edge $e=(u,v)$ in $G$ with $u,v\in V_i$: 
	  	\begin{align*}
		  y'_{(u,\out_e)} = y'_{(\out_e, \inn_e)} = y'_{(\inn_e, v)} = {x^*_{(u,v)}}\left(1- \frac{1}{x^*(\delta^+(V_i))}\right)\,.
     	\end{align*}
	\item for each edge $e=(u,v)$ in $G$ with $u\in V_i, v\in V_j$ where $i\neq j$:
	  \begin{align*}
		y'_{(\out_e,\inn_e)}  &= x^*_{(u,v)}\,,  \\
			y'_{(A_i, \out_e)}& = \frac{x^*_{(u,v)}}{ x^*(\delta^+(V_i))}\,, \qquad y'_{(u,\out_e)} = x^*_{(u,v)} \left(1- \frac{1}{x^*(\delta^+(V_i))}\right)\,, \\	
			y'_{(\inn_e,A_j)}& = \frac{x^*_{(u,v)}}{ x^*(\delta^+(V_j))}\,, \qquad y'_{(\inn_e, v)} = x^*_{(u,v)} \left(1- \frac{1}{x^*(\delta^+(V_j))}\right)\,. \\
	  \end{align*}
  \end{enumerate}
  Basically, $y'$ is defined so that a fraction $1/x^*(\delta^+(V_i))$ of the
  flow crossing the cut $(V_i, V \setminus V_i)$ goes through $A_i$. As
  $x^*(\delta^+(V_i))\geq 1$ we have that $y'$ is nonnegative. It is also
  immediate from the definition of $y'$ that it satisfies flow conservation and
  the degree bounds of the claim: the in- and out-flow of a vertex $v\in V_i$
  is $\left( 1- \frac{1}{x^*(\delta^+(V_i))} \right) x^*(\delta^+(v))$; the in-
  and out-flow of an auxiliary vertex $A_i$ is $1$ by design; and the in- and
  out-flow of $O_e$ and $I_e$ for $e=(u,v)$ is  $(1- 1/x^*(\delta^+(V_i)))x^*_e$ if $u,v \in V_i$ for some $i=1,\ldots, k$ and $x^*_e$ otherwise. As mentioned above, the existence of
  fractional circulation $y'$ implies that we can  also find, in polynomial
  time, an integral circulation $y$ with the required properties. 
\end{proof}

Having found an integral circulation $y$ as in the above claim, we now obtain
the Eulerian subset $F$ of edges. Initially, the set $F$ contains $y_{(\out_e,
\inn_e)}$ multiplicities of each edge $e$ in $G$. Note that with respect to
this edge set, in each partition $V_i$, either all vertices in $V_i$ are
 balanced (each vertex's in-degree equals its out-degree) or there exist exactly one
vertex $u$ so that $|\delta^+_{F}(u)| - |\delta^-_{F}(u)| = -1$ and one vertex $v$
so that $|\delta^+_{F}(v)| - |\delta^-_{F}(v)| = 1$. Specifically, let $u$ be the
head of the unique edge $e$ such that $y_{(\inn_e, A_i)}=1$ and let $v$ be the
tail of the unique edge $e'$ so that $y_{(A_i, \out_{e'})} = 1$. If $u=v$ then
all vertices in $V_i$ are balanced. Otherwise $u$ is so that $|\delta^+_{F}(u)|
- |\delta^-_{F}(u)| = -1$ and $v$ is so that $|\delta^+_{F}(v)| - |\delta^-_{F}(v)|
= 1$. In that case, we add a simple path from $u$ to $v$ to make the in-degrees
and out-degrees of these vertices balanced. As the graph induced by $V_i$ is
strongly connected, we can select the path so that it only visits vertices in $V_i$.
Therefore, we only increase the degree  of vertices in $V_i$ by at most $1$.
Hence, after repeating this operation for each partition $V_i$, we have that
$F$ is a Eulerian subset of edges and $|\delta_{F}(\delta^+(v))| \leq
y(\delta^+(v))+1 \leq \lceil x^*(\delta^+(v)) \rceil + 1$ for all $v\in V$.
Finally, we have $|\delta^+_{F}(V_i)| \geq 1$ for each $i=1,\ldots, k$ because
$y_{A_i, \out_e} = 1$ (and therefore $y_{\out_e, \inn_e}\geq 1$) for one edge
$e\in \delta^+(V_i)$. We have thus given a polynomial time algorithm that finds
a Eulerian subset $F$ satisfying the properties of~\eqref{eq:cycprop}, which,
as discussed above, implies that it is a  $3$-light algorithm for Node-Weighted
Local-Connectivity ATSP.
\end{proof}

\section{From Local to Global Connectivity}
\label{sec:LocalToGlobal}
In this section, we prove that if there is an $\alpha$-light algorithm for
Local-Connectivity ATSP, then there
exists an algorithm for ATSP with an approximation guarantee of $O(\alpha)$. 
The main theorem can be stated as follows.

\begin{theorem}
  Let $\cA$ be an algorithm for Local-Connectivity ATSP and  consider an ATSP
  instance $G=(V,E, w)$.  If $\cA$ is $\alpha$-light on $G$,  there exists a tour of $G$ with value at
  most $5\alpha\lb(V)$. 
  Moreover, for any $\varepsilon >0$,  a tour of value at most
  $(9+\varepsilon)\alpha\lb(V)$ can be found in time polynomial in the number
  $n=|V|$ of vertices, in $1/\varepsilon$,  and in the running time of $\cA$.
  \label{thm:LoctoGlo}
\end{theorem}

Throughout this section, we let $G=(V,E,w)$ and $\cA$ be fixed as in the
statement of the theorem. The proof of the theorem is by giving an algorithm
that uses $\cA$ as a subroutine. We first give the non-polynomial algorithm in
Section~\ref{sec:ExistDescAlg} (with the better guarantee) followed by Section~\ref{sec:polyalg} where we
modify the arguments so that we also  efficiently find a tour (with slightly worse guarantee).

\subsection{Existence of a Good Tour}
\label{sec:ExistDescAlg}
Before describing the (non-polynomial) algorithm, we need to introduce the concept of Eulerian
partition. 
We say that graphs $H_1 = (V_1, E_1), H_2=(V_2,E_2), \ldots, H_k = (V_k, E_k)$
form a \emph{Eulerian partition} of $G$ if the vertex sets $V_1, \ldots, V_k$
form a partition of $V$ and each $H_i$ is a connected Eulerian graph where
$E_i$ is a multisubset of $E$. It is an \emph{$\beta$-light Eulerian partition}
if in addition
\begin{align*}
w(H_i) & \leq \beta\cdot \lb(H_i) \qquad \mbox{for } i =1, \ldots, k.
\end{align*}
Our goal is to find a $5\alpha$-light Eulerian partition that only
consists of a single component, i.e., a $5\alpha$-approximate solution to the
ATSP instance $G$ with respect to the Held-Karp relaxation.

The idea of the algorithm is to start with a Eulerian partition and then
iteratively merge/connect these connected components into a single connected
component by adding  (cheap) Eulerian subsets of edges. Note that, since we
will only add Eulerian subsets, the algorithm always maintains that the
connected components are Eulerian. 

The \emph{state of the algorithm} is described by a Eulerian multiset $E^*$
that contains the multiplicities of the edges that the algorithm has picked.

\paragraph{Initialization}
The algorithm starts with a $2\alpha$-light Eulerian partition $H^*_1= (V^*_1, E^*_1), \ldots, H^*_k = (V^*_k, E^*_k)$ that maximizes the lexicographic order of
\begin{align}
  \label{eq:lexord}
\langle \lb(H^*_1), \lb(H^*_2), \ldots, \lb(H^*_k) \rangle.
\end{align}
As the lexicographic order is maximized, the Eulerian partitions are ordered so
that $\lb(H^*_1) \geq \lb(H^*_2) \geq \cdots \geq \lb(H^*_k)$. For simplicity,
we assume that these inequalities are strict (which is w.l.o.g. by 
breaking ties arbitrarily but consistently).
The set $E^*$ is initialized so that it contains the edges of the Eulerian
partitions, i.e., $E^* = E^*_1 \cup E^*_2 \cup \cdots \cup E^*_k$. 

During the
execution of the algorithm we will also use the following concept. For a
connected subgraph $\tilde G = (\tilde V, \tilde E)$ of $G$,  let
$\low(\tilde G)$ denote the Eulerian partition $H^*_i$ of lowest index
$i$ that intersects $\tilde G$\footnote{Equivalently, it is the set $
H^*_i$ maximizing $\lb(H^*_i)$ over all sets in the Eulerian partition that intersect
$\tilde G$.}. That is, 
\begin{align*}
  \low(\tilde G) =  H^*_{\min\{i: V^*_i \cap \tilde V \neq \emptyset\}}. 
\end{align*}
Note that after initialization, the connected components in $\con(E^*)$ are exactly the
subgraphs $H^*_1, \ldots, H^*_k$. This means that $H^*_i = \low(\tilde G)$
for exactly one component $\tilde G \in \con(E^*)$.  
Moreover, as the algorithm will only add edges, each $H^*_i$ will be in at
most one component throughout the execution. 

\begin{remark}
  \label{rem:init}
  The main difference in the polynomial time algorithm is the initialization
  since we do not know how to find a $2\alpha$-light Eulerian partition
  that maximizes the lexicographic order in polynomial time. 
  Indeed, it is consistent with our knowledge that $2\alpha$ (even $2$) is an upper bound on the
  integrality gap and, in that case,  such an algorithm would  always find a
  tour.
\end{remark}
\begin{remark}
  For intuition, let us mention that the reason for starting with a Eulerian partition that maximizes the lexicographic order is that we will use the following properties to bound the weight of the total tour:
  \begin{enumerate}
    \item A connected Eulerian subgraph $H$ of $G$ with $w(H)\leq 2\alpha \lb(H)$ has $\lb(H) \leq \lb(\low(H))$.
    \item For any disjoint connected Eulerian subgraphs $H_1, H_2, \ldots, H_\ell$ of $G$ with $\low(H_j) = H^*_i$ and $w(H_j) \leq \alpha \lb(H_j)$ for $j=1, \ldots, \ell$, we have
  \begin{align*}
	\sum_{j=1}^\ell \lb(H_j) \le 2 \lb(H^*_i).
  \end{align*}
  \end{enumerate}
  These bounds will be used to bound the weight of the edges added in the merge procedure. Their proofs are easy and can be found in the analysis (see the proofs of Claim~\ref{claim:boundlex1} and Claim~\ref{claim:boundlex2}).
\end{remark}

\paragraph{Merge procedure}

The algorithm repeats the following ``merge procedure'' until $\con(E^*)$ contains a single connected component.
The components in $\con(E^*)$ partition the vertex set and each component is strongly connected as it is Eulerian (since $E^*$ is a Eulerian subset of edges).
 The algorithm can therefore use $\cA$ to find a Eulerian multisubset $F$ of  $E$ such that   
\begin{itemize}
  \item[(i)] $|\delta^+_{F}(\tilde V)|  \geq 1 \qquad \mbox{ for all } (\tilde V, \tilde E) \in \con(E^*)$; and
  \item[(ii)]  for each $\tilde G \in \con(F)$ we have $w(\tilde G) \leq \alpha \lb(\tilde G)$.
\end{itemize}
Note that $\cA$ is guaranteed to find such a set since it is assumed to be an $\alpha$-light algorithm for Local-Connectivity ATSP on $G$. 
Furthermore, we
may assume that no connected component in $\con(F)$ is completely contained in
a connected component in $\con(E^*)$ (except for the trivial components formed by
singletons). Indeed, the edges of such a component can
safely be removed from $F$ and we have a new (smaller) multiset that satisfies the
above conditions.
Having selected $F$, we now proceed to explain the ``update phase'':

\begin{enumerate}[label=U\arabic*:, ref=U\arabic* ]\itemsep0mm
	\item \label{en:u1}  Let $X = \emptyset$.  
	\item \label{en:u2}Select the component $\tilde G = (\tilde V, \tilde E) \in \con(E^* \cup F \cup X)$ that \emph{minimizes}	$\lb(\low(\tilde G))$.
  \item \label{en:u3}If there exists a cycle $C = (V_C, E_C)$ in $G$  of weight $w(C) \leq \alpha \lb(\low(\tilde G))$
	that connects $\tilde G$ to another component in $\con(E^* \cup F \cup X)$, then add $E_C$ to $X$ and
	repeat from Step~\ref{en:u2}.
  \item \label{en:u4}Otherwise, update $E^*$ by adding the ``new'' edges in $\tilde E$, i.e.,  $E^* \leftarrow E^* \cup (\tilde E \cap F) \cup (\tilde E \cap X)$. 
	
\end{enumerate}

Some comments about the update of $E^*$ are in order. We emphasize that we do \emph{not}
add all edges of $F \cup X$ to $E^*$. Instead, we only add those new edges that belong to  
 the component $\tilde G$ selected in the final iteration of the update phase.
As $\tilde G$ is a connected component in $\con(E^* \cup F
\cup X)$, $F$ and $X$ are Eulerian subsets of edges, we have that $E^*$
remains Eulerian after the update. This finishes the description of the merging procedure and the algorithm (see also the example below). 

\begin{figure}[h]
\centering
\newcommand{\base}{
\draw[fill=gray!20] plot [smooth cycle,tension=0.7] coordinates { (-1,-1) 
  (1,-1) (1.4,1) (-1,1)};
  \draw[gray] (-1,-1) edge[-, bend right = 10 ] (0,0);
  \draw[gray] (1,-1) edge[-, bend left = 20] (0,0);
  \draw[gray] (1.4,1) edge[-, bend right = 15 ] (0,0);
  \draw[gray] (-1,1) edge[-, bend left = 10] (0,0);
  \node at (0.1, 0.8) {$H^*_{10}$};
  \node at (-0.8, 0.0) {$H^*_9$};
  \node at (0.1, -0.8) {$H^*_7$};
  \node at (0.9, 0.0) {$H^*_6$};

  \begin{scope}[xshift=4.5cm]
	\draw[fill=gray!20] plot [smooth cycle,tension=0.7] coordinates { (-0.7,-0.7) 
  		(0.5,-0.7) (1, 0) (0.8,0.7) (-0.8,0.8) (-1.2, 0)};
	\node at (0.0, 0.0) {$H^*_3$};
  \end{scope}

  \begin{scope}[xshift=8cm]
	\draw[fill=gray!20] plot [smooth cycle,tension=0.7] coordinates { (-0.7,-0.7) 
  		(0.5,-0.7) (1, 0) (0.8,0.7) (-0.8,0.8) (-0.6, 0)};
  	\draw[gray] (-0.6,0) edge[-, bend left =20] (1,0);
	\node at (0.2, -0.35) {$H^*_8$};
	\node at (0.2, 0.5) {$H^*_5$};
  \end{scope}

  \begin{scope}[yshift = -3.5cm]
  \begin{scope}[xshift=0cm]
	\draw[fill=gray!20] plot [smooth cycle,tension=0.7] coordinates { (-0.7,-0.7) 
  		(0.5,-0.7) (1, 0) (0.8,0.7) (0, 1) (-0.8,0.8) (-1.2, 0)};
	\node at (0.0, 0.1) {$H^*_4$};
  \end{scope}
  \begin{scope}[xshift=4.5cm]
	\draw[fill=gray!20] plot [smooth cycle,tension=0.7] coordinates { (-0.7,-0.7) 
  		(0.5,-0.7)  (0.8,0.7) (0, 1.2) (-0.8,0.8) (-1.2, 0)};
	\node at (-0.1, 0.1) {$H^*_2$};
  \end{scope}
  \begin{scope}[xshift=8cm, yshift=0.3cm]
	\draw[fill=gray!20] plot [smooth cycle,tension=0.7] coordinates { (-1,-1) 
  		(0.5,-1)  (0.8,0.7) (0, 1.2) (-1,1) (-1.4, 0) };
	\node at (-0.2, -0.0) {$H^*_1$};
  \end{scope}
  \end{scope}
}
\begin{tikzpicture}
\base
\draw[thick, blue] plot [smooth cycle, tension = 1] coordinates {(-0.4,-0.5) (-0.4, -2.5) (0.7, -2.6) (1, -0.5) };
\draw[ultra thick, blue] plot [smooth cycle, tension = 1] coordinates {(5.6,0.0) (5.6, 0.7) (7.5, 0.7) (7.5, -0.2) };
\draw[thick, blue] plot [smooth cycle, tension = 1] coordinates {(5.2, -3.7) (5.2,-3) (7.1, -3.7) (7.1, -2.5) };
\draw[thick, red, dashed] plot [smooth cycle, tension = 1] coordinates {(4.2, -2.5) (3.5,-2.5)  (0.2, -0.2) (1.5, 0.1) };
\end{tikzpicture}
\caption{An illustration of the merge procedure. Blue  (solid) cycles depict $F$
and the red (dashed) cycle depicts $X$ after one iteration of the update
phase. The thick cycle represents the edges that this merge procedure would add to
$E^*$.}
\label{fig:algo}
\end{figure}
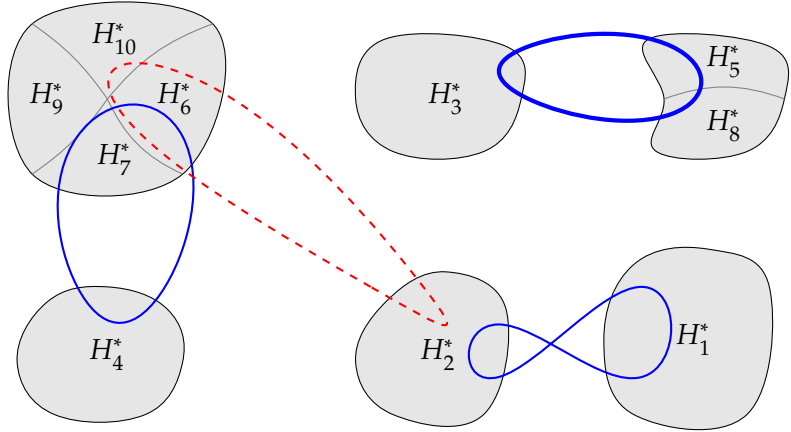

\begin{example}
  In Figure~\ref{fig:algo}, we have that, at the start of a merging step, $\con(E^*)$ consists of $6$ components containing
$\{H^*_6, H^*_7, H^*_9, H^*_{10}\}, \{H^*_3\}, \{H^*_5, H^*_8\}, \{H^*_4\}, \{H^*_2\},$ and  $\{H^*_1\}$.
The blue (solid) cycles depict the connected Eulerian components of the edge
set $F$. First, we set $X =\emptyset$ and the algorithm selects the
component $\tilde G$ in $\con(E^* \cup F \cup X)$ that minimizes $\lb(\low(\tilde G))$ or, equivalently, that maximizes $\min \{i: H^*_i \mbox{ intersects } \tilde G\}$. In
this example, it would be the left most of the three components in $\con(E^* \cup
F)$ with $\low(\tilde G) = H^*_4$. The algorithm now tries to connect this component
to another component by adding a cycle with weight at most
$\alpha \lb(H^*_4)$. The red (dashed) cycle corresponds to such a cycle and its edge set  is
added to $X$. In the next iteration, the algorithm considers the two
components in $\con(E^* \cup F \cup X)$.  The smallest component (with
respect to $\lb(\low(\tilde G))$)  is the one that contains $H^*_3,
H^*_5$, and $H^*_8$. Now suppose that there is no cycle of weight at
most $\alpha \lb(H^*_3)$ that connects this component to another component.
Then the set $E^*$ is updated by adding those edges of $F\cup X$ that
belong to this component (depicted by the thick cycle).
\end{example}

\subsubsection{Analysis}

\paragraph{Termination} We show that the algorithm terminates by arguing that
the update phase terminates with fewer connected components and the merge procedure is therefore 
 repeated at most $k\leq n$ times. 

\begin{lemma}
  \label{lem:termination}
The update phase terminates in polynomial time and decreases the number of connected components in $\con(E^*)$. 
 \end{lemma}
 \begin{proof}
   First, observe that each single step of the update phase can be implemented
   in polynomial time. The only nontrivial part is Step~\ref{en:u3} which can
   be implemented as follows: for each edge $(u,v)\in \delta^+(\tilde V)$ consider the cycle consisting of $(u,v)$ and 
   a shortest path from $v$ to $u$. Moreover, the whole update phase terminates
   in polynomial time because  each time the if-condition of Step~\ref{en:u3}
   is satisfied, we add a cycle to $X$ that decreases the number of connected
   components in $\con(E^* \cup F \cup X)$.  The if-condition of
   Step~\ref{en:u3} can therefore be satisfied at most $ k \leq n$ times.

   We proceed by proving that at termination the update phase decreases the
   number of connected components in $\con(E^*)$. Consider when the algorithm
   reaches Step~\ref{en:u4}. In that case it has selected a component $\tilde
   G=(\tilde V, \tilde E) \in \con(E^* \cup F \cup X)$. Note that $\tilde G
   \not \in \con(E^*)$ because the edge set $F$ crosses each cut defined by
   the vertex sets of the connected components in $\con(E^*)$. Therefore when
   the algorithm updates $E^*$ by adding all  the edges $(F \cup X) \cap
   \tilde E$ it decreases the number of components in $\con(E^*)$ by at least
   one. 
 \end{proof}

 \paragraph{Performance Guarantee}

 To analyze the performance guarantee we shall split our analysis into two
 parts. Note that when one execution of the merge procedure terminates (Step~\ref{en:u4}) we add edge set
 $(F\cap \tilde E) \cup (X \cap \tilde E)$ to our solution. We shall analyze the
 contribution of these two sets $F \cap \tilde E$ and $X \cap \tilde E$ separately.
 More formally, suppose that the algorithm does $T$ repetitions of
 the merge procedure. Let $\tilde G_1= (\tilde V_1, \tilde E_1), \tilde G_2 = (\tilde V_2, \tilde E_2), \ldots, \tilde G_T= (\tilde V_T, \tilde E_T)$, $F_1, F_2, \ldots, F_T$, and
 $X_1, X_2, \ldots, X_T$ denote the selected components, the edge set
 $F$, and the edge set $X$, respectively, at the end of each repetition.
  To simplify notation, we denote the edges added to $E^*$ in the $t$:th
 repetition by $\tilde F_{t} = F_t \cap \tilde E_t$ and $\tilde X_{t} =
 X_t \cap \tilde E_t$.
 
 With this notation, we proceed to bound the total weight of the solution by
 \begin{align*}
   \underbrace{w\left(\cup_{t=1}^T \tilde F_{t}\right)}_{\leq 2\alpha \lb(V)\mbox{\scriptsize\  by Lemma~\ref{lem:bound2}}} 
   + \underbrace{w\left(\cup_{t=1}^T \tilde X_{t}\right)}_{\leq \alpha \lb(V)\mbox{\scriptsize\ by Lemma~\ref{lem:bound1}}} 
   + {\sum_{i=1}^k w(H^*_i)}
   \leq 5\alpha \lb(V) \mbox{ as claimed in Theorem~\ref{thm:LoctoGlo}.}
 \end{align*}
 Here we used that $\sum_{i=1}^k w(H^*_i) \leq 2\alpha \lb(V)$ since $H^*_1, \ldots, H^*_k$ is a $2\alpha$-light Eulerian partition. It remains to prove Lemmas~\ref{lem:bound1} and~\ref{lem:bound2}.

 \begin{lemma} 
   \label{lem:bound1}We have 
   $w\left(\cup_{t=1}^T \tilde X_{t} \right)  \leq \alpha \lb(V).$ 
\end{lemma}

\begin{proof}
  Note that $\tilde X_t$ consists of a subset of the cycles added to $X_t$
  in Step~\ref{en:u3} of the update phase. Specifically, those cycles contained  in the
  connected component $\tilde G_t$ selected at Step~\ref{en:u2} in the last iteration
  of the update phase during the  $t$:th repetition of the merge procedure. We can therefore decompose
  $\cup_{t=1}^T \tilde X_{t}$ into cycles $C_1 = (V_1, E_1), C_2 = (V_2,
  E_2), \ldots, C_c = (V_c, E_c)$  indexed in the order they were added by the
  algorithm.  
  When $C_j$ was selected in Step~\ref{en:u3} of the update phase, it satisfied the following two properties:
  \begin{itemize}
	\item[(i)] it connected the component $\tilde G$  selected in Step~\ref{en:u2} with at least one other component $\tilde G'$ such that $\lb(\low(\tilde G')) > \lb(\low(\tilde G))$; and

	\item[(ii)] it had weight $w(C_j)$ at most $\alpha\lb(\low(\tilde G))$. 
  \end{itemize}
  In this case, we say that $C_j$ marks $\low(\tilde G)$.

  We claim that at most one cycle in $C_1, C_2, \ldots, C_c$  marks
  each $H^*_1, H^*_2, \ldots, H^*_k$. To see this, consider the
  first cycle $C_j$ that marks $H^*_i$ (if any). By (i) above,
  when $C_j$ was added, it connected two components $\tilde G$ and
  $\tilde G'$ such that $\lb(\low(\tilde G')) > \lb(\low(\tilde G))$ where $\low(\tilde G) = H^*_i$.   As the algorithm
  only adds edges, $\tilde G$ and $\tilde G'$ will remain connected throughout the
  execution of the algorithm. Therefore, by the definition of $\low$ and by
  the fact that $\lb(\low(\tilde G')) > \lb(\low(\tilde G))$, we have that a component $\tilde G''$ appearing later in the algorithm  always
  has $\low(\tilde G'') \neq H_i^*$. Hence, no other cycle marks $H_i^*$.

 The bound now follows from that  at most one  cycle
 marks each $H^*_i$ and such a cycle has weight at most $\alpha
 \lb(H^*_i)$.
\end{proof}
We complete the analysis of the performance guarantee with the following lemma.
\begin{lemma}
  \label{lem:bound2}
  We have
  $w\left(\cup_{t=1}^T \tilde F_{t} \right)  \leq  2 \alpha \lb(V)$. 
\end{lemma}

\begin{proof}
  Consider the $t$:th repetition of the merge procedure. The edge set $\tilde F_t$ is Eulerian but not necessarily
  connected. Let  $\cF^{t}$ denote the set of the Eulerian subgraphs corresponding to 
  the connected components in $\con(\tilde F_t)$ where we disregard the
  trivial components that only consist of a single vertex.  Further, 
  partition $\cF^t$ into  $\cF^t_1, \cF^t_2, \ldots, \cF^t_k$ where $\cF^t_i$ contains those  Eulerian subgraphs in
  $\cF^t$ that intersect $H^*_i$ and do not intersect any of the subgraphs
  $H^*_1, H^*_2, \ldots, H^*_{i-1}$. That is, 
  \begin{align*}
	\cF^t_i = \{H\in \cF^t: \low(H) = H_i^*\}. 
  \end{align*}
  Note that the total weight of $\tilde F_t$, $w(\tilde F_t)$, equals $w(\cF^t)= \sum_{i=1}^k w(\cF^t_i)$. We bound the weight of $\cF^t$ by considering each $\cF^t_i$ separately. 
  We start by two simple claims that follow from that each $H \in \cF^t$
  satisfies $w(H) \leq \alpha \lb(H)$ (since $\cA$ is an $\alpha$-light
  algorithm) and the choice of $H^*_1, \ldots, H^*_k$ to maximize the
  lexicographic order of~\eqref{eq:lexord}. We remark that the proofs of the 
  following claims are the only arguments that use the fact that the
  lexicographic order was maximized.

  \begin{claim}
	\label{claim:boundlex1}
	For $H\in \cF^t_i$, we have $\lb(H) \leq \lb(\low(H)) = \lb(H_i^*)$.
  \end{claim}
  \begin{proof}
	Inequality $\lb(H) > \lb(H_i^*)$ together with the fact that $w(H) \leq \alpha \lb(H) \leq 2\alpha \lb(H)$ would contradict that $H^*_1, \ldots, H^*_k$
	was chosen to maximize the lexicographic order of~\eqref{eq:lexord}. Indeed, in that case, a
	$2\alpha$-light Eulerian partition of higher lexicographic order
	would  be $H^*_1, H^*_2, \ldots, H_{i-1}^*, H$ and the remaining vertices
	(as trivial singleton components) that do not belong to any of these
	Eulerian subgraphs. 
  \end{proof}

  \begin{claim}
	\label{claim:boundlex2}
	We have $\lb(\cF^t_i) \leq  2\lb(H^*_i)$.
  \end{claim}
  \begin{proof}
    Suppose toward contradiction that $\lb(\cF^t_i) >  2\lb(H^*_i)$.  Let
    $\cF^t_i =\{H_1, H_2, \ldots, H_\ell\}$ and define $H^*$ to be the Eulerian
    graph obtained by taking the union of the graphs $H_i^*$ and $H_1, \ldots,
    H_\ell$.  Consider the Eulerian partition $H^*_1, \ldots, H^*_{i-1}, H^*$
    and the remaining vertices (as trivial singleton components) that do not
    belong to any of these Eulerian subgraphs. We have $\lb(H^*) > \lb(H_i^*)$
    and therefore the lexicographic value of this Eulerian partition is larger
    than the lexicographic value of $H^*_1, \ldots, H^*_k$. This is
    a contradiction if it is also a $2\alpha$-light Eulerian partition, i.e.,
    if $\frac{w(H^*)}{\lb(H^*)} \leq 2\alpha$. 
    
    Therefore, we must have $w(H^*) > 2\alpha \lb(H^*)$.  By the facts
    that $w(H_j) \leq \alpha \lb(H_j)$ (since $\cA$ is an $\alpha$-light
    algorithm) and that $H^*_1, \ldots, H_k^*$ is a $2\alpha$-light Eulerian
    partition, $$ w(H^*)= w(H^*_i) + \sum_{j=1}^\ell w(H_j) \leq 2\alpha
    \lb(H^*_i) + \sum_{j=1}^\ell \alpha \lb(H_j)\quad\mbox{ and }\quad
    \lb(H^*)\geq \sum_{j=1}^\ell \lb(H_j).  $$ 
    These inequalities together with $w(H^*)
    >  2\alpha \lb(H^*)$ imply $\lb(\cF^t_i) = \sum_{j=1}^\ell \lb(H_j)
    \leq 2 \lb(H^*_i)$.
  \end{proof}
  
  Using the above claim, we can write $w\left(\cup_{t=1}^T \tilde F_{t} \right)$ as
  $$
  \sum_{t=1}^T \sum_{i=1}^k  w(\cF^t_i) 
  \leq 	\alpha \sum_{t=1}^T \sum_{i=1}^k  \lb(\cF^t_i) 
  = \alpha \sum_{i=1}^k \sum_{t: \cF_i^t \neq \emptyset} \lb(\cF_i^t) 
  \leq 2\alpha \sum_{i=1}^k \sum_{t: \cF_i^t \neq \emptyset} \lb(H^*_i). 
  $$
  We complete the proof of the lemma by using Claim~\ref{claim:boundlex1} to prove that
  $\cF^t_i$ is non-empty for at most one repetition $t$ of the merge procedure.
  Suppose toward contradiction that  there exist $1\leq t_0< t_1 \leq T$ so
  that both $\cF^{t_0}_i \neq \emptyset$ and $\cF^{t_1}_i \neq \emptyset$. In
  the $t_0$:th repetition of the merge procedure, $H^*_i$ is contained in the
  subgraph $\tilde G_{t_0}$ since otherwise no edges incident to $H^*_i$ would
  have been added to $E^*$. Therefore $\lb(\low(\tilde G_{t_0})) \geq
  \lb(H^*_i)$.
  Now consider a Eulerian
  subgraph $H\in \cF^{t_1}_i$. First, we cannot have that $H$ is contained in
  the component $\tilde G_{t_0}$ since each (nontrivial) component of $F$ is assumed to
  not be contained in any component of $\con(E^*)$. Second, by Claim~\ref{claim:boundlex1},
  we have  $w(H) \leq \alpha \lb(H) \leq \alpha \lb(H_i^*)$.

  In short, $H$ is a Eulerian subgraph  that connects $\tilde G_{t_0}$
  to another component and it has weight at most $\alpha \lb(\low (\tilde
  G_{t_0}))$. As $H$ is Eulerian, it can be decomposed into cycles. One
  of these cycles, say $C$,  connects $\tilde G_{t_0}$ to another component and 
  \begin{align}
	w(C) \leq w(H) \leq \alpha \lb(H^*_i) \leq \alpha \lb(\low(\tilde G_{t_0})).
	\label{eq:cycleweight}
  \end{align}
  In other words, there exists a cycle $C$ that, in the $t_0$:th repetition of
  the merge procedure, satisfied the if-condition of Step~\ref{en:u3}, which
  contradicts the fact that Step~\ref{en:u4} was reached when component $\tilde
  G_{t_0}$ was selected. 
\end{proof}

\subsection{Polynomial Time Algorithm}
\label{sec:polyalg}
In this section we describe how to modify the arguments in
Section~\ref{sec:ExistDescAlg} to obtain an algorithm that runs in time
polynomial in the number $n$ of vertices, in $1/\varepsilon$, and in the running time of $\cA$.

By Lemma~\ref{lem:termination}, the update phase can be implemented in polynomial
time in $n$. Therefore, the merge procedure described in
Section~\ref{sec:ExistDescAlg} runs in time polynomial in $n$ and in the
running time of $\cA$.  The problem is the initialization: as mentioned in
Remark~\ref{rem:init}, it seems difficult to find a polynomial time algorithm for
finding a $2\alpha$-light Eulerian partition $H_1^*, \ldots, H_k^*$ that maximizes the lexicographic
order of
\begin{align*}
  \langle \lb(H^*_1), \lb(H^*_2), \ldots, \lb(H^*_k) \rangle.
\end{align*}

We overcome this obstacle by first identifying the properties that we actually
use from selecting the Eulerian partition as above. We then show that we can
obtain a Eulerian partition that satisfies these properties in polynomial
time. 

As mentioned in the analysis in Section~\ref{sec:ExistDescAlg}, the only place
where we use that the Eulerian partition  maximizes the lexicographic order
of~\eqref{eq:lexord} is in the proof of Lemma~\ref{lem:bound2}. Specifically,
it is used in the proofs of Claims~\ref{claim:boundlex1}
and~\ref{claim:boundlex2}. Instead of proving these claims,  we shall simply
concentrate on finding a Eulerian partition that satisfies a relaxed variant
of them (formalized in the lemma below, see Condition~\eqref{eq:relax}). The
claimed polynomial time algorithm is then obtained by first proving that
a slight modification of the merge procedure returns a tour of value at most
$(9\alpha +2\varepsilon)\lb(V)$ if Condition~\eqref{eq:relax} holds, and then we show that a
Eulerian partition satisfying this condition can be found in time polynomial in
$n$ and in the running time of $\cA$. We start by describing the 
modification to the merge procedure. 

\paragraph{Modified merge procedure} The only modification to the merge
procedure in Section~\ref{sec:ExistDescAlg} is that we change the update phase
by  relaxing the condition of the if-statement in Step~\ref{en:u3} from $w(C)
\leq  \alpha \lb(\low(\tilde G))$ to $w(C) \leq \alpha( 3 \lb(\low(\tilde
G)) + \varepsilon \lb(V)/n)$. In other words, Step~\ref{en:u3} is replaced by 
\begin{enumerate}[label=U\arabic*':, ref=U\arabic*']
	 \setcounter{enumi}{2}
  \item \label{en:uu3}If there exists a cycle $C = (V_C, E_C)$ in $G$  of weight $w(C) \leq \alpha( 3 \lb(\low(\tilde G)) + \varepsilon \lb(V)/n)$
	that connects $\tilde G$ to another component in $\con(E^* \cup F \cup X)$, then add $E_C$ to $X$ and
	repeat from Step~\ref{en:u2}.
\end{enumerate}
Clearly the modified merge procedure still runs in time polynomial in $n$ and
in the running time of $\cA$. Moreover, we show that if Condition~\eqref{eq:relax} holds then the returned tour will have weight $O(\alpha)$. Recall from
Section~\ref{sec:ExistDescAlg}  that $\tilde F_t$ denotes the subset of $F$ and $\tilde X_t$
  denotes the subset of $X$ that were added in the $t$:th repetition of the
  (modified) merge procedure.  Furthermore, we define (as in the previous section)
$\cF^t_i = \{H\in \con(\tilde F_t): \low(H) = H_i^*$ and $H$ is a nontrivial component,
i.e., $H$ contains more than one vertex$\}$.
\begin{lemma}
  \label{lem:polybound}
  Assume that the algorithm is initialized with a $3\alpha$-light Eulerian
  partition $H^*_1, H^*_2, \ldots, H^*_k$ so that,  in each repetition $t$ of the
  modified merge procedure, we add a subset $\tilde F_t$  such that 
  \begin{align}
    \lb(\cF^t_i) \leq 3\lb(H_i^*) + \frac{\varepsilon \lb(V)}{n} \qquad \mbox{for }i=1,2, \ldots, k. 
	\label{eq:relax}
  \end{align}
  Then the returned tour has weight at most $(9+2\epsilon)\alpha\lb(V)$.
\end{lemma}
Let us comment on the above statement before giving its proof. The reason that
we use a $3\alpha$-light Eulerian partition (instead of one that is
$2\alpha$-light) is that it leads to a better constant when balancing the
parameters. We also remark that~\eqref{eq:relax}  is a relaxation of the bound
of Claim~\ref{claim:boundlex2} from $\lb(\cF^t_i) < 2 \lb(H_i^*)$ to
$\lb(\cF^t_i) \leq 3 \lb(H_i^*)+ \varepsilon \lb(V)/n$; and it also implies a relaxed version of
Claim~\ref{claim:boundlex1}: from $\lb(H) \leq \lb(H_i^*)$ to $\lb(H) \leq 3\lb(H_i^*) + \varepsilon \lb(V)/n$.
It is because of this relaxed bound that we modified the if-condition of the
update phase (by relaxing it by the same amount) which will be apparent in
the proof.
\begin{proof}

  As in the analysis of the performance guarantee in
  Section~\ref{sec:ExistDescAlg}, we can write the weight of the returned tour
  as 
  \begin{align*}
    w\left(\cup_{t=1}^T \tilde F_{t}\right) + 
    w\left(\cup_{t=1}^T \tilde X_{t}\right) + 
    \sum_{i=1}^k w(H^*_i).  
  \end{align*} 
 
  To bound $w\left(\cup_{t=1}^T \tilde X_{t}\right)$, we  observe that proof of
  Lemma~\ref{lem:bound1} generalizes verbatim except that the weight of a cycle
  $C$ that marks $H_i^*$ is now bounded by $\alpha(3\lb(H_i^*) + \varepsilon \lb(V)/n)$ instead of by
  $\alpha \lb(H_i^*)$ (because of the relaxation of the bound in the
  if-condition of the update phase). Hence, $w\left(\cup_{t=1}^T \tilde
  X_{t}\right) \leq \sum_{i=1}^k \alpha(3\lb(H_i^*) + \varepsilon \lb(V)/n) \leq (3+\epsilon)\alpha \lb(V)$ because $k\leq n$. 

  We proceed to bound $w\left(\cup_{t=1}^T \tilde F_{t}\right)$. Using the
  same arguments as in the proof of Lemma~\ref{lem:bound2},
  \begin{align*}
    w\left(\cup_{t=1}^T \tilde F_{t}\right) 
    \leq  \alpha \sum_{i=1}^k \sum_{t: \cF_i^t \neq \emptyset} \lb(\cF_i^t) 
    \leq \alpha \sum_{i=1}^k \sum_{t: \cF_i^t \neq \emptyset} \left(3\lb(H_i^*) + \varepsilon \lb(V)/n\right)
  \end{align*}
  where, for the last inequality, we used the assumption of the lemma. Now we
  apply exactly the same arguments as in the end of the proof of
  Lemma~\ref{lem:bound2} to prove that $\cF^t_i$ is non-empty for at most one
  repetition $t$ of the merge procedure. The only difference, is
  that~\eqref{eq:cycleweight} should be  replaced by
  \begin{align*}
    w(C) \leq w(H) \leq \alpha (3\lb(H_i^*)+ \varepsilon \lb(V)/n) \leq \alpha (3\lb(\low(\tilde G_{t_0}))+ \varepsilon \lb(V)/n)
  \end{align*}
  (because~\eqref{eq:relax} can be seen as a relaxed version of
  Claim~\ref{claim:boundlex1}). However, as we also updated the bound in the
  if-condition, the argument that $C$ would satisfy the if-condition of
  Step~\ref{en:uu3} is still valid.  Hence, we conclude that $\cF_i^t$ is non-empty in at most one repetition and therefore
  \begin{align*}
     w\left(\cup_{t=1}^T \tilde F_{t}\right) 
    \leq \alpha \sum_{i=1}^k \sum_{t: \cF_i^t \neq \emptyset} \left(3\lb(H_i^*) + \varepsilon \lb(V)/n\right)
    \leq (3+\varepsilon) \alpha \lb(V).
  \end{align*}

  By the above bounds and  since $H_1^*, H_2^*, \ldots, H_k^*$ is
  a $3\alpha$-light Eulerian partition, we have that the weight of the returned
  tour is
  \begin{align*}
    w\left(\cup_{t=1}^T \tilde F_{t}\right) + 
    w\left(\cup_{t=1}^T \tilde X_{t}\right) + 
    \sum_{i=1}^k w(H^*_i)
    &\leq (3+\varepsilon)\alpha\lb(V) + (3+\varepsilon)\alpha\lb(V) + 3\alpha\lb(V)\\
    &= (9+2\varepsilon)\alpha\lb(V).
  \end{align*} 
\end{proof}

\paragraph{Finding a good Eulerian partition in polynomial time}
By the above lemma, it is sufficient to find a $3\alpha$-light Eulerian
partition so that Condition~\eqref{eq:relax} holds during the execution of the modified
merge procedure. However, how can we do it in polynomial time? We do as
follows. First, we select the trivial $3\alpha$-light Eulerian partition where
each subgraph is only a single vertex.  Then we run the modified merge
procedure and, in each repetition, we verify that Condition~\eqref{eq:relax}
holds. Note that this condition is easy to verify in time polynomial in $n$. If
it holds until we return  a tour, then we know by Lemma~\ref{lem:polybound}
that the tour has weight at most $(9+2\varepsilon)\alpha\lb(V)$. If it does not hold during
one repetition, then we will restart the algorithm with  a new
$3\alpha$-light Eulerian partition that we find using the following lemma. We
continue in this manner until the merge procedure executes without violating
Condition~\eqref{eq:relax} and therefore it returns a tour of weight at most
$(9\alpha+2\varepsilon) \lb(V)$.

\begin{lemma}
  \label{lem:polytime}
  Suppose that repetition $t$ of the (modified) merge procedure violates
  Condition~\eqref{eq:relax} when run starting from a $3\alpha$-light Eulerian
  partition $H_1^*, H_2^*, \ldots, H_k^*$. Then we can, in time polynomial in
  $n$, find a new $3\alpha$-light Eulerian partition $\hat H_1^*, \hat H_2^*, \ldots,
  \hat H^*_{\hat k}$ so that 
  \begin{align}
    \sum_{j=1}^{\hat k} \lb(\hat H^*_j)^2 - \sum_{j=1}^k \lb(H^*_j)^2 \geq \frac{\varepsilon^2}{3n^2}\lb(V)^2.
	\label{eq:potential}
  \end{align}
\end{lemma}
Note that the above lemma implies that we will reinitialize (in polynomial
time) the Eulerian partition at most $3n^2/\varepsilon^2$ times because any
Eulerian partition $H_1^*, \ldots, H_k^*$ has $\sum_{i=1}^k \lb(H_i^*)^2 \leq
\lb(V)^2$. As each execution of the merge procedure takes time polynomial in
$n$ and in the running time of $\cA$, we can therefore find a tour of weight at
most $(9+2\varepsilon)\alpha \lb(V)= (9+\varepsilon') \alpha \lb(V)$ in the
time claimed in Theorem~\ref{thm:LoctoGlo}, i.e., polynomial in $n$,
$1/\varepsilon'$, and in the running time of $\cA$. It remains to prove the
lemma. 
\begin{proof}
  Since the $t$:th repetition of the merge procedure violates Condition~\eqref{eq:relax},  there is an $1\leq i \leq k$ such that
  \begin{align*}
    \lb(\cF^t_i) > 3\lb(H_i^*) + \frac{\varepsilon}{n} \lb(V).
  \end{align*}
  We shall use this fact to construct a new $3\alpha$-light Eulerian partition
  consisting of a new Eulerian subgraph $H^*$ together with a subset of
  $\{H_1^*, H_2^*, \ldots, H_{k}^*\}$ containing those subgraphs that do not
  intersect $H^*$  and finally the vertices (as trivial singleton components)
  that do not belong to any of these Eulerian subgraphs. We need to define the
  Eulerian subgraph $H^*$. Let $I \subseteq
  \{1,2, \ldots, k\}$ be the indices of those Eulerian subgraphs of 
  $H_1^*, \ldots, H_k^*$ that intersect the  vertices in $\cF_i^t$.
  Note that, by definition, we have $i\in I$ and $j\geq i$ for all $j \in I$.
  We shall construct the graph $H^*$ iteratively. Initially, we let $H^*$ be
  the connected Eulerian subgraph obtained by taking the union of $\cF^t_i$ and
  $H_i^*$. This is a connected Eulerian subgraph as each Eulerian subgraph in
  $\cF^t_i$ intersects $H_i^*$ and $H_i^*$ is a connected Eulerian subgraph.

  The careful reader can observe that up to now $H^*$ is defined in the same
  way as in the proof of Claim~\ref{claim:boundlex2}. However, in order to
  satisfy~\eqref{eq:potential} we shall add more of the Eulerian subgraphs in
  $\{H^*_j\}_{j\in I}$ to $H^*$. Specifically, we would like to add
  $\{H^*_j\}_{j\in I'}$, where $I' \subseteq I \setminus\{i\}$ is selected so as
  to  
  maximize $\lb(H^*)$ (because we wish to increase the ``potential''
  in~\eqref{eq:potential})  subject to that $w(H^*) \leq 3\alpha \lb(H^*)$
  (because the new Eulerian partition should be $3\alpha$-light).

  To see that $w(H^*) \leq 3\alpha \lb(H^*)$ implies that the new Eulerian
  partition is $3\alpha$-light, recall that the new Eulerian partition consists
  of $H^*$, the Eulerian subgraphs $\{H_j^*\}_{j\not \in I}$, and the vertices
  that do not belong to any of these Eulerian subgraphs.  By the definition of
  $I$, no $H_j^*$ with $j \not \in I$ intersects $H^*$. As $H_1^*, \ldots,
  H_k^*$ are disjoint, it follows that the new Eulerian partition consists of
  disjoint subgraphs. Moreover, each $H_j^*$ satisfies $w(H_j^*) \leq 3\alpha
  \lb(H_j^*)$ since the Eulerian partition we started with is $3\alpha$-light.
  Hence, the new Eulerian partition is $3\alpha$-light if $w(H^*) \leq 3\alpha
  \lb(H^*)$. Inequality~\eqref{eq:cap} is thus a sufficient condition for the new Eulerian
  partition to be $3\alpha$-light. We remark that the condition trivially holds
  for $I' = \emptyset$ because $\lb(\cF_i^t) > 3\lb(H_i^*) + \varepsilon
  \lb(V)/n$. 
  \begin{claim}
    We have $w(H^*) \leq 3\alpha\lb(H^*)$ if 
    \begin{align}
      \sum_{j\in I'} \lb(H_j^* \cap \cF^t_i) \leq \frac{2}{3}\lb(\cF^t_i)  - \lb(H_i^*\cap \cF_i^t).
      \label{eq:cap}
    \end{align}
  \end{claim}
  \begin{proof}
    We have
    \begin{align*}
      w(H^*) = w(\cF_i^t) + w(H_i^*)+ \sum_{j\in I'} w(H_j^*) 
      \leq \alpha \lb(\cF_i^t) + 3\alpha \lb(H_i^*) + 3\alpha \sum_{j\in I'} \lb(H_j^*),
    \end{align*}
    where the inequality follows from that $\cF_i^t$ was selected by the
    $\alpha$-light algorithm $\cA$ and $H_1^*, \ldots, H_k^*$ is a
    $3\alpha$-light Eulerian partition. Moreover,
    \begin{align*}
      \lb(H^*)  = \lb(\cF^t_i) + \lb(H_i^* \setminus \cF^t_i) + \sum_{j\in I'} \lb(H_j^* \setminus \cF_i^t).
    \end{align*}
    Hence, we have, by rearranging terms and using $\lb(H_j^*) - \lb(H_j^*
    \setminus \cF^t_i) = \lb(H_j^* \cap \cF^t_i)$, that $w(H^*) \leq 3\alpha
    \lb(H^*)$ holds if    
    \begin{align*}
      3\alpha \lb(H_i^* \cap \cF_i^t) + 3\alpha \sum_{j\in I'} \lb(H_j^* \cap \cF^t_i) \leq  2\alpha \lb(\cF_i^t). 
    \end{align*}
    The above can be simplified to 
    \begin{align*}
       \sum_{j\in I'} \lb(H_j^* \cap \cF^t_i) \leq  2\lb(\cF_i^t)/3- \lb(H_i^* \cap \cF_i^t).
    \end{align*}
  \end{proof}
  
  From the above discussion, we wish to find a subset $I' \subseteq I\setminus
  \{i\}$ that satisfies~\eqref{eq:cap} and maximizes 
  \begin{align*}
    \lb(H^*) = \lb(\cF^t_i) + \lb(H_i^* \setminus \cF^t_i) + \sum_{j\in I'}  \lb(H_i^* \setminus \cF^t_i),
  \end{align*}
  where only the last term depends on the selection of $I'$.
  We interpret this as a knapsack problem that, for each $j\in I\setminus
  \{i\}$, has an item of size $s_j = \lb(H_j^* \cap \cF^t_i)$ and profit $p_j
  = \lb(H_j^* \setminus \cF_i^t)$; the  capacity $U$ of the knapsack is $\frac{2}{3} \lb(\cF^t_i)
  - \lb(H_i^* \cap \cF^t_i)$, i.e.,  the right-hand-side of~\eqref{eq:cap}. We
  solve this knapsack problem and obtain $I'$ as follows:
  \begin{enumerate}
    \item Find an optimal  extreme point solution $z^*$ to the standard
      linear programming relaxation of the knapsack problem:
      \begin{align*}
        \textrm{maximize}& \sum_{j\in I\setminus \{i\}} z_j p_j \\
        \textrm{subject to}& \sum_{j\in I\setminus \{i\}} z_j s_j \leq U, \\
        & \quad 0\leq z_j \leq 1 \qquad \mbox{for all } j\in I\setminus \{i\}.
      \end{align*}
    \item As the above relaxation has only one constraint (apart from the
      boundary constraints), the extreme point $z^*$ has at most one variable with a fractional
      value. We obtain an integral solution (i.e., a packing) by simply dropping the
      fractionally packed item. That is, we let $I' = \{j\in I\setminus \{i\}: z^*_j
      = 1\}$.  
  \end{enumerate}
   The running time of the above procedure is dominated by the time it takes to solve the linear
   program. This can be done very efficiently by solving the fractional
   knapsack problem with the greedy algorithm (or, for the purpose here, use
   any general polynomial time algorithm for linear programming). We can
   therefore obtain $I'$ and the new Eulerian partition in time polynomial in
   $|I| \leq n$ as stated in lemma.

  It remains to prove~\eqref{eq:potential}. Let us first bound the profit of our ``knapsack solution'' $I'$.
  \begin{claim}
    \label{claim:knapsack}
    We have $\sum_{j\in I'} \lb(H_j^* \setminus \cF^t_i) \geq \frac{1}{3}
    \sum_{j\in I\setminus i} \lb(H_j^*\setminus \cF^t_i) - \lb(H_i^*)$.
  \end{claim}
  \begin{proof}
    By definition,
    \begin{align*}
      \sum_{j\in I'} \lb(H_j^* \setminus \cF^t_i) = \sum_{j\in I\setminus \{i\}:  z^*_j = 1}p_j
      \geq \sum_{j\in I\setminus \{i\}} z_j^* p_j - \max_{j\in I \setminus \{i\}} p_j,
    \end{align*}
    where we used that at most one item is fractionally packed in $z^*$. As
    $j\geq i$ for all $j\in I$, $\max_{j\in I\setminus \{i\}} p_j = \max_{j\in
      I\setminus \{i\}} \lb(H^*_j \setminus \cF^t_i) \leq \lb(H^*_i)$.
      To complete the proof of the claim, it is thus sufficient to prove that $z'_j =1/3$ for all $j\in I\setminus \{i\}$ is a feasible solution to the LP relaxation of the knapsack problem. Indeed, by the optimality of $z^*$, we then have
      $\sum_{j\in I \setminus \{i\}} z_j^* p_j \geq \sum_{j\in I \setminus \{i\}} z_j' p_j = \frac{1}{3} \sum_{j\in I\setminus \{i\}} \lb(H_j^* \setminus \cF^t_i)$. 

      We have that $z'$ is a feasible  solution because 
      \begin{align*}
        \frac{1}{3}  \sum_{j\in I\setminus\{i\}} \lb(H_j^* \cap \cF_i^t) \leq \frac{1}{3} \lb(\cF^t_i) 
        &\leq \left(\frac{2}{3} -\frac{\lb(H_i^* \cap \cF_i^t)}{\lb(\cF^t_i)}\right) \lb(\cF_i^t) = U, 
      \end{align*}
  where the first inequality follows from that the subgraphs $\{H_j^*\}_{j\in
  I}$ are disjoint and the second inequality follows from that
  $\lb(H_i^*)/\lb(\cF^t_i) \leq 1/3$.
  \end{proof}
  We finish the proof of the lemma by using the above claim to show  the increase of the
  ``potential'' function as stated in~\eqref{eq:potential}. By
  the definition of the new Eulerian partition (it contains $\{H_j^*\}_{j\not
  \in I}$), we have that the increase is at least
  \begin{align*}
	\lb(H^*)^2 - \sum_{j\in I} \lb(H_j^*)^2. 
  \end{align*}
  Let us concentrate on the first term:
	\begin{align*}
	  \lb(H^*)^2& = \left( \lb(\cF^t_i) + \lb(H_i^* \setminus \cF^t_i) + \sum_{j\in I'} \lb(H_j^* \setminus \cF^t_i)\right)^2 \\
	   &\geq \lb(\cF^t_i)\left( \lb(\cF^t_i) + \lb(H_i^* \setminus \cF^t_i) + \sum_{j\in I'} \lb(H_j^* \setminus \cF^t_i)\right).
  \end{align*}
  By Claim~\ref{claim:knapsack}, we have that the expression inside the parenthesis is at least
  \begin{align*}
    \lb(\cF^t_i) + \lb(H_i^* \setminus \cF^t_i) + & \frac{1}{3}\sum_{j\in I\setminus \{i\}} \lb(H_j^* \setminus \cF^t_i) - \lb(H_i^*) \\
   & \geq 
    \lb(\cF^t_i) +  \frac{1}{3}\sum_{j\in I} \lb(H_j^* \setminus \cF^t_i) - \lb(H_i^*).
  \end{align*}
  By using $\lb(H_i^*) \leq \lb(\cF^t_i)/3$, we can further lower bound this expression by
  \begin{align*}
     \frac{1}{3}\lb(\cF_i^t)  + \frac{1}{3}\left(\lb(\cF^t_i) + \sum_{j\in I} \lb(H_j^* \setminus \cF^t_i)\right) =  \frac{1}{3}\lb(\cF_i^t)  + \frac{1}{3}\sum_{j\in I} \lb(H_j^*).
  \end{align*}
  Finally, as $\lb(\cF^t_i) \geq \varepsilon \lb(V)/n$, $\lb(\cF^t_i) \geq 3\lb(H_i^*)$, and $\lb(H_j^*) \leq \lb(H_i^*)$ for all $j\in I$, we have  
  \begin{align*}
    \lb(H^*)^2 - \sum_{j\in I} \lb(H_j^*)^2 & \geq \lb(H^*)^2 - \lb(H_i^*) \sum_{j\in I} \lb(H_j^*) \\
    & \geq \lb(\cF^t_i) \left( \frac{1}{3}\lb(\cF_i^t)  + \frac{1}{3}\sum_{j\in I} \lb(H_j^*) \right) - \lb(H_i^*) \sum_{j\in I} \lb(H_j^*) \\
    &\geq \frac{\lb(\cF^t_i)^2}{3} + \frac{\lb(\cF^t_i)}{3} \sum_{j\in I} \lb(H_j^*)  - \lb(H_i^*) \sum_{j\in I} \lb(H_j^*) \\
    &\geq  \frac{\varepsilon^2}{3n^2} \lb(V)^2
  \end{align*}
  which completes the proof of Lemma~\ref{lem:polytime}.
\end{proof}

\section{Discussion and Open Problems}
\label{sec:discuss}
We gave  a new approach for approximating the asymmetric
traveling salesman problem. It is based on relaxing the global connectivity
requirements into local connectivity conditions, which is formalized as 
Local-Connectivity ATSP. We showed a rather easy $3$-light
algorithm for Local-Connectivity ATSP on shortest path metrics of node-weighted
graphs. This yields via our generic reduction a constant factor approximation
algorithm for Node-Weighted ATSP. However, we do not know any $O(1)$-light algorithm for Local-Connectivity ATSP on general metrics and, motivated by our generic reduction, we raise
the following intriguing  question:

\begin{openquestion}
Is there a $O(1)$-light algorithm for Local-Connectivity ATSP on general metrics? 
\end{openquestion}
We note that there is great flexibility in the exact choice of the lower bound
$\lb$ as noted in Remark~\ref{rem:lb}. A further generalization of our
approach is to interpret it as a primal-dual approach.  Specifically, it might
be useful to interpret the lower bound as a feasible solution of the dual of
the Held-Karp relaxation: the lower bound is then not only defined over the
vertices but over all cuts in the graph. We do not know if any of these
generalizations are useful at this point and it may be that there is a nice
$O(1)$-light algorithm for Local-Connectivity ATSP without changing the definition of $\lb$. 

By specializing the generic reduction to Node-Weighted ATSP,  it is possible to improve our bounds slightly for this case. 
Specifically, one can exploit the fact that a cycle $C$ always has $w(C) \leq
\lb(C)$ in these metrics. This allows one to change the bound in
Step~\ref{en:u3} of the update phase to be $w(C) \leq \lb(\low(\tilde G))$
instead of   $w(C) \leq \alpha \lb(\low(\tilde G))$, which in turn improves the
upper bound on the integrality gap of the Held-Karp relaxation to $4\cdot
\alpha + 1= 13$ (since $\alpha =3$ for node-weighted metrics). That said, we do
not see how to make a significant improvement in the guarantee and  it would be very interesting with a tight analysis of the integrality gap of the Held-Karp relaxation for Node-Weighted ATSP. We believe that such a result would also be very interesting even if we restrict ourselves to shortest path metrics of unweighted graphs.

Finally, let us remark that the recent progress for STSP on shortest path
metrics of unweighted graphs is not known to extend to  node-weighted graphs,
i.e., Node-Weighted STSP. Is it possible to give a $(1.5-\varepsilon)$-approximation algorithm for Node-Weighted STSP for some constant $\varepsilon>0$?  
We think that this is a very natural question that lies in between the now fairly well understood 
STSP on shortest path metrics of unweighted graphs and STSP on general metrics (i.e., edge-weighted instead of node-weighted graphs).

\subsection*{Acknowledgments}
The author is very grateful to L\'{a}szl\'{o} V\'{e}gh,  Johan H\aa stad, and Hyung-Chan An for inspiring
discussions and valuable comments that influenced this work.
We also thank Jakub Tarnawski and Jens Vygen for useful feedback on the manuscript.

This research is supported by ERC Starting Grant 335288-OptApprox.

\bibliographystyle{abbrv}
{\small \bibliography{atsp}}

\end{document}